\documentclass[12pt,fleqn]{article}
\usepackage[margin=3.5cm]{geometry}
\usepackage{lipsum} 
\usepackage{amsmath}
\usepackage{mathtools}
\usepackage{slashed}
\usepackage{url}
\usepackage{bm}
\usepackage{graphicx}
\usepackage[utf8]{inputenc}
\usepackage[english]{babel}
\usepackage{amsmath}
\usepackage{amsfonts}
\usepackage{amssymb}
\usepackage{float}
\usepackage{hyperref}
\usepackage{extarrows}
\usepackage{fullpage}
\usepackage{amsmath,amscd}
\usepackage{amsthm}
\usepackage{tikz-cd}
\theoremstyle{lemma} 
\newtheorem{lemma}{Lemma}
\theoremstyle{remark} 
\newtheorem{remark}{Remark}
\theoremstyle{proposition} 

\theoremstyle{definition} 
\newtheorem{definition}{Definition}
\begin{document}
\linespread{1.5}
\renewcommand{\baselinestretch}{1.2}
  \fontsize{13}{16}\selectfont

\title{Spin current in BF theory}
\author         {Malik  Almatwi}

\date{Department of Mathematical Science, Ritsumeikan University, 6-1-1 Matsugaoka, Otsu, Shiga 520-2102, Japan 
 \\ 
email: {malik.matwi@gmail.com}}
\maketitle
\tableofcontents

\section*{Abstract}
{In this paper}, we introduce a current which we call spin current corresponding to the variation of the matter action in BF theory with respect to the spin connection $A$ which takes values in Lie algebra $\mathfrak{so}(3,\mathbb{C})$ in self-dual formalism. For keeping the constraint $DB^i=0$ satisfied, we suggest adding a new term to the BF Lagrangian using a new field $\psi^i$ which can be used for calculating the spin current. We derive the equations of motion and discuss the solutions. We will see that the solutions of that equations do not require a specific metric on the manifold $M$, we just need to know the symmetry of the system and the information about the spin current. Finally we find the solutions in a spherical and cylindrical symmetric systems.
\\

Keywords: {BF} theory; local Lorentz symmetry; local Lorentz currents.

\section{Introduction}
The BF theory on 4-manifold $M$ is a topological theory which when includes constraints terms turns to gravity theory. The fundamental variables are 2-form $B  \in \Omega^2(M; \mathfrak{so}(3,1))$ and spin connection $\omega$ which takes values in $\mathfrak{so}(3,1)$, all derivatives are linear and applied only on $\omega$, which makes it easy for canonical formalism; finding the phase space, Hamiltonian equations, quantization,... \cite{Laurent}. This theory does not require a metric to be formulated, the metric is a derived quantity from the solutions of $B$. That gives motivation to formulate Einstein’s gravity as a theory of 2-forms rather than the metric tensors and so no pre-exiting geometrical structure is needed to obtain the gravity. Let $F(\omega) \in \Omega^2(M; \mathfrak{so}(3,1))$ be the curvature of $\omega$. The pure BF theory action is $ \int\limits_M tr({B  \wedge F (\omega)})$ which is invariant(symmetric) under local Lorentz transformation(regarded as gauge group) and under arbitrary diffeomorphisms of $M$ and does not need using a metric. The equations of motion are $F(\omega)=0$ and $d_\omega B=0$, thus $B$ defines a twisted de Rham cohomology class $[B] \in H^2_{DR}(M, \mathfrak{so}(3,1))$, and the solution of $F(\omega)=0$ is unique up to gauge and diffeomorphism transformations. There are no local degrees of freedom because the system has so much symmetry, that all solutions are locally equivalent under gauge transformation of the group $SO(3,1)$ and under diffeomorphisms of $M$. Hence the pure BF theory is a topological theory \cite{John, Aberto}. \\

For example, let $f: M \to M$ be an infinitesimal diffeomorphism generated by a vector field $v$, then $B \mapsto B+ d_\omega (i_v(B))$ for $d_ \omega B=0$. Therefore $d_\omega B \mapsto d_\omega B+ d^2_\omega (i_v(B))$, but $F(\omega)=0$ implies $d^2_\omega=0$, thus the equation of motion $d_\omega B=0$ maps to the equation of motion $d_{\omega} B=0$, therefore all equations of motion are equivalent under diffeomorphisms of $M$. But when $F(\omega)\ne 0$, then $d^2_\omega\ne0$, therefore the equation of motion $d_\omega B=0$ maps to the equation of motion $d_{\omega} B \ne0$, so there are local degrees of freedom. That relates to the fact that $d_\omega$ acts locally when $F(\omega)\ne 0$, in other words $f^*(d_\omega)\ne d_\omega$ when $F(\omega)\ne 0$, but when $F(\omega)= 0$, then $f^*(d_\omega)= d_\omega$. In QFT, using fields(not forms), the action changes(under $f$) as $\delta S=\int \partial_\mu(g_{\nu \rho}\delta x^\rho \delta S/\delta g_{\mu \nu} )$ when the equation of motion are satisfied. Therefore in order to get TQFT, it must be $\delta S/\delta g_{\mu \nu}=0$, so independent of any metric(standard or other).
\\

In constrained BF, the Lagrangian includes the constraint term $\varphi_{IJKL} B^{IJ} \wedge B^{KL}$. The traceless matrix $\varphi$ plays the role of a Lagrangian multiplier and that imposes the constraint on the 2-form $B^{IJ}$, so that its solutions are given in terms of 1-forms $e^I=e^I_\mu dx^\mu$, that is $B^{IJ}=e^I \wedge e^J$, where $I, J,...=0,1,2,3$ are Lorentz indices and $\mu, \nu,...=0,1,2,3$ are spacetime tangent indices, we regard the frame fields $e^I_\mu dx^\mu$ as gravitational fields, therefore the constrained BF theory turns to general relativity theory, the reason is that when $\varphi_{IJKL}$ is not zero, the term  $\varphi_{IJKL} B^{IJ} \wedge B^{KL}$ breaks the diffeomorphisms invariance of BF action, thus there are non-equivalent local solutions and so local degrees of freedom exist as known in general relativity in the vacuum. The problem with constrained BF theory is that the equation of motion $\delta S/ \delta B=0$ contains the non-physical variable $\varphi_{IJKL}$, but we can remove it by taking the trace of the equations($\varphi$ is traceless matrix), but also there is a problem with the trace operation, it reduces the equations to one equation which is not enough for getting a solution. For that reason we search for solutions of BF theory by using the equation $\delta S/ \delta \omega=0$(as done in this paper). In general, the equations of motion of constrained BF theory including matter give a relation between the curvature $F^{IJ}(\omega)$ and the frame fields $\Sigma^{IJ}=e^I \wedge e^J$(the Plebanski 2-form), in matrix notation, that is $F=\chi \Sigma+ \xi \bar\Sigma $, where $\chi$ and $\xi$ are symmetric matrices of scalar fields \cite{Mariano}. Therefore the problem turns to finding $\chi$ and $\xi$. Since the field $\varphi_{IJKL}$ is not a physical variable, the equations of motion of general relativity have not to include it(appendix A).
\\

In this paper, we start with definition of the spin current $J$ and discuss its conservation in BF theory including matter(for a general we do not specify a matter Lagrangian). The spin current $J$ appears in the equations of motion as a source for $d_\omega B$ by the equation $*(d_\omega B)+J=0$, where '$*$' is Hodge star operator. And in order to get $d_\omega B=0$ in this paper, we add a new term to BF Lagrangian, like $tr\left(\psi B \wedge F(\omega)\right)$, using a new field $\psi$. That can be seen as a redefinition $B \to B+\psi B$, by which the equation $*(d_\omega B)+J=0$ becomes $*(d_\omega B)+*(d_\omega(\psi B))+J=0$, so we choose $d_\omega B=0$ and get $*((d_\omega\psi) B)+J=0$. Therefore the spin current becomes a source for the field $\psi$ instead of $B$ and we get a new formula(definition) for the spin current using $\psi$, and since the spin current regards symmetry of the system, the field $\psi$ also regards that symmetry. In constraint BF theory, solving the equation $*((d_\omega\psi) B)+J=0$ is easy as we will see. We find that the equation of motion of $\psi$ is same conservation equation $D_\mu J^\mu=0$ of the spin current vector field $J$. We see that we can solve the equations of BF theory only by solving the spin current equation $\delta S/ \delta \omega=0$, $J \ne 0$ with $d_\omega B=0$ and without needing solving the equation $\delta S/ \delta B=0$ which includes the Lagrangian multiplier $\varphi_{IJKL}$(a non-physical variable), and without using a gravitational metric on $M$, we just need using the spin current and knowing the symmetry of the system. That means that we can solve the BF equations only by using the coupling term $ \int\limits_M \omega^{IJ}_\mu J^\mu_{IJ}$ which makes them easy to solve, and makes the theory similar to the gauge theory. And since $\omega^{IJ}$ is 1-form and $J_{IJ}$ is vector field, the term $ \omega^{IJ}_\mu J^\mu_{IJ}$ is naturally defined on $M$ without needing using additional structure(like a metric), thus solving the system equation using only that coupling term gives a topological theory, i.e, the theory turns to finding 1-forms and vector fields, and does not need to use a gravitational metric, similarly to Chern-Simons theory which includes the Wilson loops as a source for the gauge field. That makes it easy to solve the equations in different cases of the spin current, e.g, point charge, straight line current, circular current,... . The lines of the spin current can be described using any coordinates system, e.g, Euclidean coordinates,..., so the BF theory can be studied in any coordinates system, but in order to avoid an effect of the coordinates on the lines of spin currents, we let that coordinates be flat(not curved). And since the spin current is source for the field $\psi$, this field has singularities on the lines of that spin current. We see that our solution of $B^{IJ}$ can be always written as $e^I \wedge e^J$, so we get the gravity theory. Finally we give an example of explicit solution of the equations in spherical and cylindrical symmetric systems in static case just by finding the field $\psi$ using the spin charges $J^0_{IJ}$.

\section{Spin current in BF theory}
Let $M$ be connected oriented smooth 4-manifold and $P\to M$ be an $SO(3,1)$-principal bundle(Appendix A) with a spin connection $\omega$ which is locally a 1-form with values in $\mathfrak{so}(3,1)$ and $F \in \Omega^2(M; \mathfrak{so}(3,1)_P)$ is its curvature. The BF theory action is invariant under global and local Lorentz transformation, that gives a conserved current, we call it spin current (Appendix C). Before discussion the conservation of the spin current, we introduce the self-dual formalism.

\begin{definition}
The self-dual projection is a homomorphism 
\[
\mathfrak{so}(3,1)_P  = P \times_{SO(3,1)} \mathfrak{so}(3,1) \to  \mathfrak{so}(3,\mathbb{C})_P  =P \times_{SO(3,1)} \mathfrak{so}(3,\mathbb{C})
\]
defined by
\[
\left( {P_{IJ}^i } \right):\mathfrak{so}(3,1){\hookrightarrow}\mathfrak{so}(3,1)_\mathbb{C}  = \mathfrak{so}(3,\mathbb{C}) \oplus \mathfrak{so}(3,\mathbb{C}) \to \mathfrak{so}(3,\mathbb{C}),
\]
with using the matrices \cite{Carl}
\begin{equation}\label{eq:selfd}
P^i_{IJ}=\frac{1}{2}\varepsilon ^i{_{jk}}, \text{ for } I=i, J=j, \text{ and } P^i_{0j}=-P^i_{j0}=-\frac{i}{2}\delta^i_j, \text{ for } I=0, J=j \ne 0.
\end{equation}
We used $P \times _{SO(3,1)} \mathfrak{so}(3,1) = {{\left( {P \times \mathfrak{so}(3,1)} \right)} \mathord{\left/
 {\vphantom {{\left( {P \times \mathfrak{so}(3,1)} \right)} {SO(3,1)}}} \right.
 \kern-\nulldelimiterspace} {SO(3,1)}}$, which is locally isomorphic to $U_\alpha \times  \mathfrak{so}(3,1)$ for open sets $\{U_\alpha \}$ of $M$.
\\

That self-dual projection relates to the fact that the complexified Lie algebra of $SO(3,1)$ has the decomposition $\mathfrak{so} (3,1)_{\mathbb{C}}=\mathfrak{so}(3,{\mathbb{C}}) \oplus \mathfrak{so}(3, {\mathbb{C}})$ \cite{Herfray}. The new connection is locally an $\mathfrak{so}(3,\mathbb{C})$-valued 1-form $A$ on $M$ whose components are
\begin{equation}\label{eq:62}
A^{i}_\mu=P^i_{IJ}\omega^{IJ}_\mu=\frac{1}{2} \varepsilon{^i}_{jk} \omega^{jk}_\mu   -i\omega^{0i}_\mu, \quad i =1,2,3,
\end{equation}
and its curvature is
\begin{equation}\label{eq:63}
 F^i(A)=P^i_{IJ}F^{IJ}(\omega)=dA^i+\varepsilon{^i}_{jk} A^j\wedge A^k.
\end{equation}
The two form $B^{IJ}$ is mapped to $B^{i}=P^i_{IJ}B^{IJ}$. The covariant derivative $D_\mu   = \nabla _\mu   +\omega _\mu ^{IJ}$ acting on sections of $TM \otimes \mathfrak{so}(3,1,\mathbb{C})_P$ becomes $D_\mu  = \nabla _\mu   +A _\mu$, with $A _\mu^{ij}=\varepsilon{^{ij}}_{k} A^k_\mu$, where $\nabla _\mu$ is the affine connection on $TM$. 
\\

Using the new variables we can write the Lagrangian of matter(without specifying matter fields) $L_{matter}(e^I, \omega^{IJ})$ as $L_{matter}(B^{i}, A^{i}, \bar B^i, \bar A^i)$, where $\bar B^i$ and $\bar{ A}^i$(anti-selfdual representation) are the complex conjugation of $B^i$ and $A^i$. The Urbantke formula (equation \ref{eq:ab2}, appendix B) writes the metric $g_{\mu \nu}$ using only the constrained $B^{i}$ without using the constrained $\bar B^{i}$. And the self-dual connection ${ A}^i$ is compatible with $B^{i}$ via $d_{A} B^{i}=0$, while the anti-self-dual connection ${\bar A}^i$ is compatible with $\bar B^{i}$ via $d_{\bar A} \bar B^{i}=0$. By that we may suppose 
\begin{equation}\label{eq:ab1}
\frac{\delta }{{\delta \bar{B}^i }}L_{matter}  = 0,\quad \frac{\delta }{{\delta \bar{A}^i }}L_{matter}  = 0,
\end{equation}
or just writing $L_{matter}(B^{i}, A^{i})$.
\\

\end{definition}

\begin{definition}
Let $A$ be be the self-dual connection on the $\mathfrak{so}(3, \mathbb{C})$-bundle $\mathfrak{so}(3, \mathbb{C})_P \to M $. Let $L_{matter}$ be the Lagrangian of matter fields on $M$. Then the spin current $J_{ i}^\mu$ is defined to be
\[
J_{ i}^\mu=   \frac{\delta }{{\delta A_\mu ^i }}L_{matter}.
\]
\end{definition}
The matter action $S_{matter}$ is required to be invariant under any infinitesimal local Lorentz transformation $\omega_\mu ^{ IJ}\mapsto \omega_\mu ^{ IJ} + D_\mu  \Lambda^{ IJ}$ for infinitesimal transformation parameter $\Lambda^{ IJ} \in \Omega^0(M; \mathfrak{so}(3,1)_P)$. Now we assume $S_{matter}$ has this property. Then we have the following.
\begin{lemma}
The spin current $J_{ i}^\mu$ given by $J_i^\mu   = \frac{\delta }{{\delta A_\mu ^i }}L_{matter} $ in gravity theory is conserved \cite{Michael}.
\end{lemma}
\begin{proof}
Since $S_{matter}$ is invariant under infinitesimal gauge transformation $\Lambda^{IJ}$, it is invariant under (we may suggest the condition (\ref{eq:ab1}))
\[
A_\mu ^{ i}\mapsto A_\mu ^{ i} + D_\mu  \Lambda^{i} \quad for \quad \Lambda^{i}=P^i_{IJ} \Lambda^{ IJ}  \in \Omega^0(M; \mathfrak{so}(3, \mathbb{C})_P). 
\]
The variation
\begin{equation*}
\begin{split}
S_{matter} &\left( {A^i  + D\Lambda ^i } \right) - S_{matter} \left( {A^i } \right) = \int\limits_M d^4 x{\left( {D_\mu  \Lambda ^i } \right)\frac{\delta }{{\delta A_\mu ^i }}} L_{matter}  \\ 
 & =  - \int\limits_M d^4 x {\Lambda ^i D_\mu  \left( {\frac{\delta }{{\delta A_\mu ^i }}L_{matter} } \right)} + \int\limits_{\partial M} {d^3 x^\mu \Lambda ^i \frac{\delta }{{\delta A_\mu ^i }}L_{matter}   } 
\end{split}
\end{equation*}
vanishes for arbitrary $\Lambda ^i$ only when $D_\mu \left( \frac{\delta }{{\delta A_\mu ^i }}L_{matter}\right) =0$, where we let $\Lambda ^i$ vanish on the boundary $\partial M$. Thus the current $J_i^\mu = \frac{\delta }{{\delta A_\mu ^i }} L_{matter}$ is conserved. Actually the previous calculation based on the idea that ${ A}^i$ and ${\bar A}^i$ transform independently under infinitesimal local Lorentz transformation $\omega_\mu ^{ IJ}\mapsto \omega_\mu ^{ IJ} + D_\mu  \Lambda^{ IJ}$, therefore there is another current that associates with the connection ${\bar A}^i$ when the matter Lagrangian depends also on ${\bar A}^i$.
\end{proof}
Same thing we find for the general relativity action, by using the variables $(\Sigma^i, A^i)$, we obtain the equation
\[
\int\limits_M d^4 x{\left( {D_\mu  \Lambda ^i } \right)\frac{\delta }{{\delta A_\mu ^i }}} S_{GR}=0 \Rightarrow  -\int\limits_M d^4 x{ \Lambda ^i D_\mu \frac{\delta }{{\delta A_\mu ^i }}} S_{GR}=0.
\]
In $3+1$ decomposition of the space-time manifold $M=\Sigma \times \mathbb{R}$, let $\Sigma_t$ be space-like slice of constant time $t$, with the coordinates $(x^a)_{a=1,2,3}$, let $0$ be time index. In Hamilton–Jacobi system, by using the variables $(E_i^a, A_a^i, A_0^i)$ on the slice of constant time $\Sigma_t$, that equation becomes 
\begin{equation}\label{eq:a15}
\begin{split}
&-\int\limits_{\Sigma \times \mathbb{R}} d^4 x \Lambda ^i D_a \left( \frac{\delta }{{\delta A_a ^i }}S_{GR}\right)- \int\limits_{\Sigma \times \mathbb{R}} d^4 x \Lambda ^i D_0 \left( \frac{\delta }{{\delta A_0 ^i }}S_{GR}\right) \\
&=-constant \times \int\limits_{\Sigma \times \mathbb{R}} d^4 x \Lambda ^i \left(D_a E^{a}_i +  D_0(D_a E^{a}_i)\right)=0, 
\end{split}
\end{equation}
which is satisfied when $D_a E^{a}_i=0$, where $E^{a}_i$ is conjugate momentum to $ A_a^i$, and we used the relations $E^{a}_i=constant \times\frac{\delta }{{\delta A_a ^i }}S_{GR}$ and $D_a E^{a}_i=constant \times\frac{\delta }{{\delta A_0 ^i }}S_{GR}$  \cite{Carl, Gen, Carlo}. The equation $D_a E^{a}_i=0$ is satisfied in BF theory(appendix A).
\\

\begin{remark}
We note that the current $J_i^\mu$ is similar to the currents in Yang-Mills theory of the gauge fields, we see this clearly when we regard the connection $A_\mu ^i$ as a gauge field, by that the current $J_i^\mu$ relates to the local Lorentz invariance (local symmetry). The metric $g_{\mu \nu} = e^I_\mu e^J_\nu \eta_{IJ}$ is invariant under arbitrary local Lorentz transformations, like $e^I_\mu(x) \mapsto U^I{_J}(x)e^J_\mu(x)$, for $U(x) \in SO(3,1)$, therefore the local Lorentz symmetry is an internal degree of freedom.
\\
\end{remark}

\begin{definition}
The action of BF theory including matter(without cosmological constant) on $SO(3,1)$-principal bundle $P \to M$ is defined to be \cite{Lee}
\[
 S = S_{topological}  + S_{constraints} + S_{matter},
\]
with
\[
S_{topological}  = \int\limits_M {B_i  \wedge F^i(A) }, \quad\text{and}\quad S_{constraints}  =\frac{1}{2} \int\limits_M {\varphi _{ij} B^i  \wedge B^j } ,
\]
where $\varphi\in \Gamma\left(M; \rm{End}(\mathfrak{so}(3, \mathbb{C})_P)\right)$ is traceless matrix of scalar fields $\varphi_{ij}$, actually we do not require it to be symmetric since we will add a new term to BF Lagrangian(see the discussion below the equation (\ref{eq:a2})). The connection $A$ on the Lie algebra bundle $\mathfrak{so}(3,\mathbb{C})_P$ which is locally a 1-form with values in $\mathfrak{so}(3,\mathbb{C})$ and its curvature $F(A)\in\Omega^2\left(M; \mathfrak{so}(3, \mathbb{C})_P\right)$ are defined in the equations (\ref{eq:62}) and (\ref{eq:63}). The index contraction is done by using $\delta_{ij}$, the Killing form on $\mathfrak{so}(3, \mathbb{C})$.
\end{definition}
Hence
\begin{equation}\label{eq:51}
S = \int\limits_M {\left( {B_i  \wedge F^i (A) + \frac{1}{2}\varphi_{ij} B^i  \wedge B^j } \right)} + S_{matter} .  
\end{equation}
Since the matrix $\varphi$ is traceless, we can write $\varphi_{ij}=m_{ij}- (m_{11}+m_{22}+m_{33})\delta_{ij}/3$, for some not traceless matrix $(m_{ij})$. The variation of the action with respect to $m_{ij}$ produces a quadratic equation in $B^i$ whose solution turns the theory into general relativity. These are
\[
B^i  \wedge B^j  = \frac{1}{3}\delta ^{ij} B_i  \wedge B^j .
\]
The solutions to this are all of the following form $B^i  = P^i _{IJ} e^I  \wedge e^J $, in which the gravitational fields $e^I_\mu$ are considered as frame fields \cite{Puzio}. Using the self-dual formula (\ref{eq:selfd}), the constrained 2-form $B^i $ is written as 
\[
B^i  = \frac{1}{2}\varepsilon ^i{_{jk}} e^j  \wedge e^j-i e^0  \wedge e^i=\Sigma^i,
\]
this is $\left. {B^i } \right|_{constrained}  = \Sigma ^i $, with using the notation $\Sigma ^i  = P^i _{IJ} e^I  \wedge e^J$.
\\

The equation of motion with respect to $B^i $ is
\[
F^i (A) + \varphi ^i {_j} B^j  + \frac{{\delta S_{matter} }}{{\delta B_i }} = 0,
\]
or
\begin{equation}\label{eq:81}
F^i (A)  =- \varphi ^i {_j } B^j  - \frac{{\delta S_{matter} }}{{\delta B_i }}.
\end{equation}
Since $F^i (A) \in\Omega^2\left(M; \mathfrak{so}(3, \mathbb{C})_P\right) $ is 2-form with values in $\mathfrak{so}(3, \mathbb{C})$, the $ \frac{{\delta  }}{{\delta B_i }}S_{matter} $ is also 2-form with values in $\mathfrak{so}(3, \mathbb{C})$.

\begin{lemma}
In constrained $B^i $, the variation $ \frac{{\delta  }}{{\delta B_i }}S_{matter} \in\Omega^2\left(M; \mathfrak{so}(3, \mathbb{C})_P\right) $ has the form
\[
\left. {\frac{{\delta S_{matter} }}{{\delta B_i }} } \right|_{constraint}    =T ^i {_j} \Sigma^j +\xi^i {_j} \bar\Sigma^j,
\]
for some matrices $T ^i {_j},  \xi^i {_j} \in \Gamma\left(M; \rm{End}(\mathfrak{so}(3, \mathbb{C})_P)\right)$, with $T ^i {_j}=T  {_j} ^i$ and $\xi^i {_j} =\bar \xi {_j} ^i$  (See appendix B, for more details).
\end{lemma}
Therefore in the vacuum we set $T^i{_j}=0$. Using this formula in the equation (\ref{eq:81}) implies 
\begin{equation*}
F^i (A)  =- \varphi ^i {_j } \Sigma^j  -T ^i {_j} \Sigma^j-\xi^i {_j} \bar\Sigma^j,
\end{equation*}
or 
\begin{equation}\label{eq:41}
F^i (A)  =\psi ^i {_j } \Sigma^j -\xi^i {_j} \bar\Sigma^j,
\end{equation}
for some matrix $\psi=-\varphi- T \in \Gamma\left(M; \rm{End}(\mathfrak{so}(3, \mathbb{C})_P)\right)$  \cite{Ingemar}. 
\\

Since $tr(\varphi)=0$, so $tr(\psi)=-tr(T)$, the equation (\ref{eq:41}) yields
\begin{equation}\label{eq:z4}
\Sigma_i^{\mu\nu} F^i_{\mu\nu}  =- tr(T),\quad \text{ for }\quad \Sigma^i_{\mu\nu} \Sigma_j^{\mu\nu}=\delta^i_j ,\text{ and } \Sigma^i_{\mu\nu} \overline\Sigma_j^{\mu\nu}=0.
\end{equation}
Thus in the vacuum $T^i {_j }=0$, we have $\Sigma_i^{\mu\nu} F^i_{\mu\nu}  =0$. We call $\Sigma_i^{\mu\nu} F^i_{\mu\nu}  $ the $tr F $. The equation (\ref{eq:z4}) does not contain the non-physical variable $\varphi$, but the problem with it is that the trace process decreases the number of equations. Therefore $\Sigma_i^{\mu\nu} F^i_{\mu\nu}  =- tr(T)$ is a condition on the solutions. The (0,2) tensor $\Sigma_j^{\mu\nu}$ is inverse of the 2-form $\Sigma^i_{\mu\nu}$ (Appendix C).
\\

The equation of motion with respect to the connection $A^i $ is 
\[
DB^i   + \frac{\delta }{{\delta A^i }}S_{matter}  = 0,
\]
or
\begin{equation}\label{eq:640}
\varepsilon ^{\mu \nu \rho \sigma } D_\nu  B_{\rho \sigma }^i +J^{\mu i } = 0.
\end{equation}

We see that we can not choose $\varepsilon ^{\mu \nu \rho \sigma } D_\nu  B_{\rho \sigma }^i =0$ when $J^{\mu i } \ne 0$. But the condition $\varepsilon ^{\mu \nu \rho \sigma } D_\nu  B_{\rho \sigma }^i = 0$ leads to the equations $D_a E^{ai}=0$(appendix A) and $De^I=0$, which makes the spin connection $\omega^{IJ}$ compatible with the gravitational field $e^I$, that makes the equations of motion easy to solve as we will see. We can get $DB^i=0$ by a redefinition of $B^i$, like $B^i \to B^i+\varepsilon ^i{_{jk}}  \psi^j B^k$, so the equation (\ref{eq:640}) yields 
\begin{equation}\label{eq:z6}
\varepsilon ^{\mu \nu \rho \sigma } \varepsilon ^i{_{jk}}  (D_\nu  \psi^j )B^k_{\rho \sigma } + J^{\mu i } = 0, \quad DB^i=0.
\end{equation}
We get this redefinition by adding a new term to the BF action (\ref{eq:51}), as done below. We see that the equations of motion of general relativity are still satisfied in BF theory(appendix A) after adding the new term, therefore there is no problem with it. In constraint BF theory, solving the equation (\ref{eq:z6}) is easy, because the field $\psi^i$ satisfies the equation $D^2 \psi^i=0$ (equation (\ref{eq:z73})) with respect to a metric that satisfies $\nabla_\rho g_{\mu \nu}=0$.
\\

By acting by $D_\mu$ on the equation (\ref{eq:640}), we get
\begin{equation*}
\varepsilon ^{\mu \nu \rho \sigma }D_\mu D_\nu  B_{\rho \sigma }^i +D_\mu J^{\mu i } = 0,
\end{equation*}
and using $D_\mu J^{\mu i } = 0$, we obtain
\begin{equation*}
\varepsilon ^{\mu \nu \rho \sigma }D_\mu D_\nu  B_{\rho \sigma }^i =\varepsilon ^{\mu \nu \rho \sigma }[D_\mu , D_\nu ] B_{\rho \sigma }^i/2= 0,
\end{equation*}
but $[D_\mu , D_\nu ]=F_{\mu \nu}(A)$, therefore
\begin{equation*}
\varepsilon ^{\mu \nu \rho \sigma }(F_{\mu \nu}(A))^i{_j} B_{\rho \sigma }^j= 0.
\end{equation*}
Then using $(F_{\mu \nu} (A))_{ij}=\varepsilon_{ijk} F^k_{\mu \nu} (A)$, implies
\begin{equation}\label{eq:641}
\varepsilon ^{\mu \nu \rho \sigma }\varepsilon_{ijk} F^j_{\mu \nu}(A) B_{\rho \sigma }^k= 0.
\end{equation}
We can regard the equation (\ref{eq:641}) as an equation of motion with respect to a new field $\psi ^i$, with the possibility of choosing $D B^i=0$ with $J^{\mu i } \ne 0$.
\\

In order to include the equations $DB^i=0$, $J^{\mu i } \ne 0$ and $D_\mu J^{\mu i } = 0$ in BF theory, we suggest the following action.
\begin{definition}
We add a new term to the BF action (\ref{eq:51}) to get
\begin{equation}\label{eq:45}
S = \int\limits_M {\left( {B_i  \wedge F^i (A) + \frac{1}{2}\varphi _{ij} B^i  \wedge B^j } \right)}  + \int\limits_M {\varepsilon_{ijk} \psi ^i B^j  \wedge F^k (A)}    + S_{matter} ,
\end{equation}
in which we have added $\int\limits_M {\varepsilon _{ijk} \psi ^i B^j  \wedge F^k (A)}$, for some vector field $\psi ^i  \in \Gamma\left(M; \mathfrak{so}(3, \mathbb{C})_P\right)$. We can get the new term by a redefinition $B^i \to B^i+\varepsilon^i{_{jk}}\psi^j B^k$ in pure BF Lagrangian $B_i  \wedge F^i (A) $. Also we can get it by a redefinition of $\varphi _{ij}$ as ${\varphi '}_{ij} B^i  = \varphi _{ij} B^i  + \varepsilon _{ijk} \psi ^i F^k(A)$, and by using $F^i  (A)= \chi ^i {_j} B^j  + \xi^i {_j} \bar B^j $ with $B^i  \wedge \bar B^j  = 0$, we obtain ${\varphi '}_{ij}  = \varphi _{ij}  + \varepsilon _{\ell ik} \psi ^\ell  \chi^k {_j}$. There is no problem with the redefinition of $\varphi _{ij}$ since they are just Lagrangian multipliers.
\end{definition}
The equation of motion of this action with respect to the field $\psi ^i$ is
\begin{equation}\label{eq:ab4}
\varepsilon_{ijk} \varepsilon ^{\mu \nu \rho \sigma } B^j_{\mu \nu}  F^k_{\rho \sigma} (A)=0,
\end{equation}
which is same equation (\ref{eq:641}), therefore the field $\psi ^i$ does not change the equations of motion. By using $(F_{\mu \nu} (A))_{ij}=\varepsilon_{ijk} F^k_{\mu \nu} (A)$, we get 
\[
 \varepsilon ^{\mu \nu \rho \sigma }  B^j_{\mu \nu} (F_{\rho \sigma} (A))^i{_j}=0,
\]
but $[D_\mu , D_\nu]=F_{\mu \nu} (A)$. Therefore
\[
 \varepsilon ^{\mu \nu \rho \sigma } [D_\rho , D_\sigma] B^i_{\mu \nu} =0,
\]
but we choose $ \varepsilon ^{\mu \nu \rho \sigma } D_\rho B^i_{\mu \nu} =0$$(DB^i=0)$ as we suggested before, thus 
\[
\frac{{\delta  }}{{\delta \psi ^i }}S_{matter}= 0
\]
is satisfied. Also the equation (\ref{eq:ab4}) is satisfied in constraint BF theory by letting the matrix $\chi _{ij}$ in $F^i  (A)= \chi ^i {_j} B^j  + \xi^i {_j} \bar B^j $ be symmetric. 
\\

The equation of motion of this action with respect to the connection $A^k$ is
\begin{equation}\label{eq:56}
DB^k  + D\left( {\varepsilon ^k {_{ij}} \psi ^i B^j } \right) + \frac{\delta }{{\delta A^k }}S_{matter}  = 0,
\end{equation}
or
\begin{equation*}
\varepsilon ^{\mu \nu \rho \sigma } D_\nu  B_{\rho \sigma }^k  + \varepsilon ^{\mu \nu \rho \sigma } D_\nu \left( {\varepsilon ^k {_{ij}} \psi ^i B_{\rho \sigma }^j } \right) + \frac{\delta }{{\delta A_\mu ^k }}S_{matter}  = 0.
\end{equation*}
We choose $DB^i=0$, which implies $D e^I=0$ in constrained BF theory, since $\left. { B^i } \right|_{constraint}   = \Sigma ^i= P^i _{IJ} e^I  \wedge e^J$. Therefore $\varepsilon ^{\mu \nu \rho \sigma } D_\nu  B_{\rho \sigma }^k  = 0$, which is equivalent to $D_\mu\left( e\Sigma^{\mu\nu i}\right)=0$ in constrained $B^i$, because of the self-duality \cite{Felix, Bennett}
\begin{equation}\label{eq:a9}
\frac{1}{{2!}}e^{-1} \varepsilon^{\mu \nu \rho \sigma } \Sigma _{\rho \sigma }^i    =\left( {*\Sigma ^i } \right)^{\mu \nu }  = \left( { - i\Sigma ^i } \right)^{\mu \nu }  =  - i\Sigma ^{\mu \nu i }, \quad e=\det(e_\mu^I),
\end{equation}
where we used the Hodge duality theory between the forms and the tensor fields, here $\Sigma^i_{\mu\nu}$ is 2-form and $\Sigma^{i\mu\nu}$ is (0,2) tensor field. The field $\Sigma^{\mu\nu}_i$ is inverse of $\Sigma^i_{\mu\nu}$, that is $\Sigma^{\mu\nu}_j \Sigma^i_{\mu\nu}=\delta^i_j$ (see the appendix C for more details). By re-scaling $\Sigma ^{\mu \nu i }$ by $e=\det(e_\mu^I)$, we write 
\begin{equation}\label{eq:z40}
D_\mu \Sigma^{\mu\nu i}=0.
\end{equation}
The remaining equation of (\ref{eq:56}) in constrained $B^i$ is 
\[
D\left( {\varepsilon ^k {_{ij}} \psi ^i \Sigma^j } \right) + \frac{\delta }{{\delta A^k }}S_{matter}  = 0,
\]
or
\[
  \varepsilon ^{\mu \nu \rho \sigma } \varepsilon ^k {_{ij}} \left( {D_\nu \psi ^i  } \right) \Sigma_{\rho \sigma }^j+ \frac{\delta }{{\delta A_\mu ^k }}S_{matter}  = 0,
\]
hence
\begin{equation}\label{eq:48}
-2i  \varepsilon ^k {_{ij}} \left( D_\nu \psi ^i\right)  \Sigma ^{\mu \nu j } + J^{\mu k }= 0,
\end{equation}
in which we used the spin current $J^\mu_ i =\delta S_{matter}/ \delta A_\mu^i$ and the condition $D_\nu \Sigma ^{\mu \nu j } =0$, also re-scaled $\Sigma ^{\mu \nu j }$ by $e=\det(e_\mu^I)$. Later we will discuss $D_\mu J^{\mu k }= 0$. Here both $\Sigma ^{\mu \nu j }$ and $J^{\mu k }$ are tensor fields. The equation (\ref{eq:48}) is same equation obtained in (\ref{eq:z6}).
\\

\begin{remark}
We note that the equation (\ref{eq:48}) is similar to the current $J^{\mu n}  = \left( {\partial ^\mu  \varphi ^i } \right)T_{ij}^n \varphi ^j $ in scalar field theory with symmetry and generators $T_{ij}^n$, we have $J^{0 n}  = \left( {\partial ^0  \varphi ^i } \right)T_{ij}^n \varphi ^j =\pi^i T_{ij}^n \varphi ^j $, where $\pi^i$ is conjugate momentum to $\varphi_i$. Like that, the equation (\ref{eq:48}) gives $ J^{0 i}=2i (D_\mu \psi_k) \varepsilon {^{ki}}_j \Sigma^{\mu0 j}$(for $\mu=a=1,2,3$), so here $D_a \psi_i$ is conjugate momentum to $\Sigma^{ a0 i}$ with noting that the indices raising in $\Sigma^{\mu\nu j}$ is done by using a metric $g_{\mu\nu}$. 
\\
\end{remark}

The equation of motion of the action (\ref{eq:45}) with respect to $B^i$ in constrained BF (like deriving the equation (\ref{eq:41})) is
\begin{equation}\label{eq:540}
F^i (A) + \varphi ^i {_j} \Sigma ^j  +\varepsilon ^i {_{jk}} \psi ^j F^k (A) + T^i {_j} \Sigma ^j  +\xi^i {_j} \bar\Sigma^j= 0.
\end{equation}
Multiplying with $\Sigma _i^{\mu \nu } $, summing over the indices and using $\Sigma _i^{\mu \nu }\bar \Sigma _{\mu \nu }^j  =0$, we obtain
\begin{equation}\label{eq:500}
\Sigma _i^{\mu \nu } F_{\mu \nu }^i  + \varphi ^i {_j} \Sigma _i^{\mu \nu } \Sigma _{\mu \nu }^j  +\varepsilon ^i {_{jk}} \psi ^j \Sigma _i^{\mu \nu } F_{\mu \nu }^k  + T^i {_j} \Sigma _i^{\mu \nu } \Sigma _{\mu \nu }^j  = 0,
\end{equation}
then using $\Sigma _i^{\mu \nu } \Sigma _{\mu \nu }^j  = \delta _i^j $ to obtain
\[
\Sigma _i^{\mu \nu } F_{\mu \nu }^i  + tr(\varphi ) +\varepsilon ^i {_{jk}} \psi ^j \Sigma _i^{\mu \nu } F_{\mu \nu }^k  + tr(T) = 0.
\]
Since $ tr(\varphi )=0$ and $\varepsilon ^i {_{jk}}  \Sigma _i^{\mu \nu } F_{\mu \nu }^k=0$(by the equations (\ref{eq:ab4}) and (\ref{eq:a9})), we obtain 
\begin{equation}\label{eq:49}
\Sigma _i^{\mu \nu } F_{\mu \nu }^i  + tr(T) =tr(F)   + tr(T) = 0.
\end{equation}
The equation (\ref{eq:540}) leads us to write $F^i (A)$ in terms of $\Sigma ^i$ and $\bar\Sigma ^i$, and since $\varepsilon ^i {_{jk}}  \Sigma _i^{\mu \nu } F_{\mu \nu }^k=0$, we can write
\begin{equation}\label{eq:50}
F^i  (A)= \chi^i {_j} \Sigma^j +{\chi'}^i {_j} \bar\Sigma^j, 
\end{equation}
for some symmetric matrix $(\chi ^{ij})$ and skew-hermitian matrix $({\chi'} ^{ij})$. Using this equation in the equation (\ref{eq:49}), we obtain
\begin{equation}\label{eq:55}
tr(\chi ^i {_j})  + tr(T^i {_j}) = 0.
\end{equation}
In addition to this relation, there is another relation between the vector field $\psi ^i$ and the symmetric matrix $\chi ^i{_j}$ when $T^i {_j} \ne 0$ and $J^{\mu i} \ne 0$, from the conservation of the current (\ref{eq:48}), $ D_\nu  J^{\nu i}  =0$, we have (for $D_\mu \Sigma ^{\mu \nu i} =0$)
\begin{equation*}
\begin{split}
\frac{i}{2}   D_\nu  J^{\nu i}  &= \left( {D_\nu  D_\mu  \psi _k } \right)\varepsilon {^{ki}} _j \Sigma ^{\mu \nu j}  = \frac{1}{2}\left(\left[ {D_\nu  ,D_\mu  } \right]\psi _k \right)\varepsilon {^{ki}} _j \Sigma ^{\mu \nu j}    \\ 
 & =  - \frac{1}{2}F_{\mu \nu }^\ell (A) \varepsilon _{\ell k} {^m} \psi _m \varepsilon {^{ki} }_j \Sigma ^{\mu \nu j}=  - \frac{1}{2}F_{\mu \nu }^\ell  (A) \varepsilon _{\ell km} \psi ^m \varepsilon ^{kij} \Sigma _j^{\mu \nu }     \\ 
&= \frac{1}{2}F_{\mu \nu }^\ell  (A) \varepsilon _{k\ell m} \psi ^m \varepsilon ^{kij} \Sigma _j^{\mu \nu } = \frac{1}{2}F_{\mu \nu }^\ell  (A) \psi ^m \left( {\delta _\ell ^i \delta _m^j  - \delta _m^i \delta _\ell ^j } \right)\Sigma _j^{\mu \nu }  \\
&= \frac{1}{2}F_{\mu \nu }^i  (A)\psi ^j \Sigma _j^{\mu \nu }  - \frac{1}{2}F_{\mu \nu }^j  (A) \psi ^i \Sigma _j^{\mu \nu },  
\end{split}
\end{equation*}
and using the equation (\ref{eq:50}), we get
\begin{equation}\label{eq:67}
\begin{split}
 \frac{i}{2}  D_\nu  J^{\nu i} &= \frac{1}{2}\chi ^i {_m} \Sigma _{\mu \nu }^m \psi ^j \Sigma _j^{\mu \nu }  - \frac{1}{2}\chi ^j {_m }\Sigma _{\mu \nu }^m \psi ^i \Sigma _j^{\mu \nu }  = \frac{1}{2}\chi ^i {_m} \psi ^j \delta _j^m  - \frac{1}{2}\chi ^j {_m} \psi ^i \delta _j^m  \\ 
 & = \frac{1}{2}\chi ^i {_j} \psi ^j  - \frac{1}{2}\psi ^i tr\left( {\chi ^i {_j }} \right) = 0 , \quad \text{ for }J^{\mu i} \ne 0.
\end{split}
\end{equation}
This is another relation between the vector field $\psi ^i$ and the symmetric matrix $\chi ^i{_j}$ in existence of matter $T^i {_j} \ne 0$ with $J^{\mu i} \ne 0$. In this case, the matrix $\chi ^i{_j}$ has to satisfy $\det\left(\chi ^i{_j}-tr( {\chi ^i {_j }}) \right)=0$ in order to get $\psi ^i \ne 0$, of course we do not need this condition in the vacuum $T^i {_j} = 0$, $J^{\mu i}=0$. 
\\

Using the equation (\ref{eq:50}) in (\ref{eq:540}), we obtain
\begin{equation*}
\chi^i {_j} \Sigma^j  + \varphi ^i {_j} \Sigma ^j  +\varepsilon ^i {_{jk}} \psi ^j \chi^k {_\ell} \Sigma^\ell  + T^i {_j} \Sigma ^j  +(...)^i {_j} \bar\Sigma^j= 0.
\end{equation*}
That yields
\begin{equation}\label{eq:a2}
\chi^{i\ell} + \varphi ^{i\ell}  +\varepsilon ^i {_{jk}} \psi ^j \chi^{k\ell}  + T^{i\ell} = 0.
\end{equation}
This equation relates to the equation of motion $\delta S/ \delta B=0$, it includes the Lagrangian multiplier $\varphi_{ij}$ which is a non-physical variable which increasing the arbitrary solutions of (\ref{eq:a2}), so increasing the local degrees of freedom. Therefore we need to find $\chi$ and $\psi$ using the other equations of motion we got before. We see that we do not require the traceless matrix $\varphi$ to be symmetric, since the third term in (\ref{eq:a2}) is not symmetric in general. The symmetric matrix $T^{ij}$ is assumed to be given using the matter Lagrangian (Appendix B), thus the total unknown variables are $3+5+8=16$ of the vector $\psi$, the symmetric matrix $\chi$(with (\ref{eq:55})) and the traceless matrix $\varphi$. The formula (\ref{eq:a2}) gives $9$ equations, therefore we have $16-9=7$ unknown variables, but when $J^{\mu i}\ne0$, they reduce to $6$ unknown variables(regarding the equation (\ref{eq:67})). But if we choose a solution for which the symmetric matrix $\chi^{ij}$ becomes diagonal, like
\begin{equation}\label{eq:a26}
\chi=\left( K^i  \delta^i_j\right)=diag(K^1, K^2, K^3), 
\end{equation}
for some scalar functions $K^1, K^2$ and $K^3$ on $M$, the unknown variables reduce to $4$ variables and to $3$ variables when $J^{\mu i}\ne0$.
\\

\begin{remark}
The field $\psi^i$ is solution of $D^\mu D_\mu \psi^i=0$ (equation (\ref{eq:z73})), so if $D_\mu v^i=0$, then $ \psi^i+v^i$ is another solution, and that makes the components $\psi^1, \psi^2$ and $\psi^3$ of the vector field $\psi^i$ independent variables, therefore we can regard them as the degrees of freedom of the system and solve the equations of motions in terms of them. We note that $\psi^i \mapsto \psi^i+v^i$ ($D v^i=0$) does not change the current $J^{\nu i} =2i \left( {  D_\mu  \psi _k } \right)\varepsilon {^{ki}} _j \Sigma ^{\mu \nu j} $. 
\\
\end{remark}
The Bianchi identity $D F^i=0$ implies $\left(D\chi^{i}{_j}\right) \wedge \Sigma^{j}=0$ (for $D \Sigma^{i}=0$), hence $(dK^i) \delta^i{_j}\wedge\Sigma^{j}=0$, where we used $D\delta^{ij}=0$, with using the covariant derivative $Dv^i  = dv^i  + \varepsilon ^i {_{jk}} A^j v^k $. Therefore we obtain
\begin{equation}\label{eq:a16}
\varepsilon ^{\mu \nu \rho \sigma } (\partial _\nu  K^i )\Sigma _{\rho \sigma }^i  = 0,\quad  \text{for}\quad \varepsilon^{0123}=-\varepsilon_{0123}=1.
\end{equation}
In $3+1$ decomposition of the space-time manifold $M=\Sigma \times \mathbb{R}$, let $\Sigma_t$ be the space-like slice of constant time $t$ with the coordinates $(x^a)_{a=1,2,3}$ (and $0$ is time index). The equation $\varepsilon^{\mu \nu \rho \sigma} D_\mu \Sigma_{ \nu \rho}^i=0$ ($DB^i=0$) decomposes to two equations, 
\begin{equation}\label{eq:a19}
D_a E^{a  i}=0  \quad\text{ and  }\quad  \varepsilon^{abc} D_b  B_c^{ i}=0, 
\end{equation}
in which we introduce the vector field $E^i$ and the 1-form $B^i$, 
\begin{equation}\label{eq:a24}
E^{a  i}=\epsilon^{0abc} \Sigma^{ i}_{bc} /2=\varepsilon^{abc} \Sigma^{ i}_{bc} /2, \quad\text{ and  } \quad B_c^{ i}=\Sigma^{ i}_{0c},
\end{equation}
on the space-like slice $\Sigma_t$(the field $E^{a  i}$ is conjugate to the connection $A^i_a$). The covariant derivative on $TM \otimes \mathfrak{so}(3, \mathbb{C})_P$ is $D_\mu=\nabla _\mu+ A _\mu$, for $A_\mu^{ij}=\varepsilon^{ij}{_k} A_\mu^{k}$, therefore on the 3d surface $\Sigma_t$, it becomes $D_a=\nabla _a+ A _a$.
\\

The equation (\ref{eq:a16}) decomposes to (for $\partial _0 K^i=0$)
\begin{equation}\label{eq:a17}
\begin{split}
& \varepsilon ^{0abc} (\partial _a K^i )\Sigma _{bc}^i  = (\partial _a K^i )\varepsilon ^{abc} \Sigma _{bc}^i  = 2(\partial _a K^i )E^{ai}  = 0, \\  
and&\\
 &\varepsilon ^{abc0} (\partial _b K^i )\Sigma _{c0}^i  =- \varepsilon ^{abc} (\partial _b K^i )B_c^i  = 0. 
\end{split} 
\end{equation}
We can solve them by writing (for non-zero curvature $F^i(A)$)
\begin{equation*}
 E^{ai}  = \frac{1}{2}\varepsilon ^{abc} (\partial _b K^i )r_c^i, \quad  \quad B_c^i  = (\partial _c K^i )u^i ,
\end{equation*}
for some $r^i \in\Omega^1\left(M; \mathfrak{so}(3, \mathbb{C})_P\right) $ and $u ^i  \in \Gamma\left(M; \mathfrak{so}(3, \mathbb{C})_P\right)$. The functions $K^i$ are scalars, the indices are just for distinguishing each from the others. Thus we get the solutions
\begin{equation}\label{eq:a18}
\Sigma ^i_{ab}  =-\Sigma ^i_{ba}  = (\partial _{[a} K^i )r_{b]}^i, \quad \text{ and  } \quad \Sigma ^i_{0a} =-\Sigma ^i_{a0} =  (\partial _a K^i )u^i.
\end{equation}
The equation (\ref{eq:a19}) implies $D r^i=0$ and $D u^i=0$.
\\

In the static case $J^{a }_i=0$(zero current) with $J^0_i\ne 0$(non-zero charge), the spin current formula $J^{\nu }_k=2i\varepsilon  _{kij} \left( D_\mu \psi ^i\right)  \Sigma ^{\mu \nu j }$, equation (\ref{eq:48}), decomposes to two equations
\begin{equation}\label{eq:a20}
\begin{split}
&J^{b }_k=2i\varepsilon_{kij} \left( D_a \psi ^i\right)  \Sigma ^{ab j }=0,   \\
and&\\
&J^{0 }_k=2i \varepsilon  _{kij} \left( D_a \psi ^i\right)  \Sigma ^{a 0 j }\ne 0.
\end{split} 
\end{equation}
We can solve the first equation in terms of $D_a \psi ^i$ by writing 
\[
\Sigma^{abi}=-\Sigma^{bai}=f v^{[a} D^{b]} \psi ^i, 
\]
for some vector $v \in \Gamma\left(M; T \Sigma\right)$ that satisfies $v^a D_a \psi ^i=0$, and $f$ is scalar function on $M$. We can include $f$ in $v$, so we just write $\Sigma^{abi}=v^{[a} D^{b]} \psi ^i$. We note that we have $D^a \psi ^i=g^{ab} D_b \psi ^i$ without needing specifying the used metric $g_{ab}$(the solution is satisfied by using an arbitrary metric). If $J^{a }_k \ne 0$, we let $\Sigma^{abi}=v^{[a} D^{b]} G ^i$ with some field $G^i\ne \psi^i$. 
\\

Regarding the second equation of (\ref{eq:a20}), when $J^{0 i}=0$, we get the solution $ \Sigma ^{0a  i }=-\Sigma ^{ a0 i }=f g^{ab} D_b \psi ^i$. And when $J^{0 i}\ne0$, we let $ \Sigma ^{0a i }=f g^{ab} D_b \xi ^i$ for some vector field $\xi^i \ne \psi ^i$, and without needing specifying the used metric $g_{ab}$ because the solution is satisfied by using an arbitrary metric. Using the equation (\ref{eq:z40}), $D_\mu \Sigma^{\mu\nu i} =0$, and for $\nu=0$, we have $D_a \Sigma^{a 0 i}=0$. Therefore the solution $ \Sigma ^{0a  i }=-\Sigma ^{ a0 i }=f g^{ab} D_b \xi ^i$ implies 
\begin{equation}\label{eq:z41}
g^{ab} D_a D_b \xi ^i=0,
\end{equation}
for constant $f$ and with using $\nabla_a g_{bc}=0$ which defines affine connection on $\Sigma_t$ in terms of the arbitrary metric $g_{ab}$. Therefore in the vacuum, where $\xi^i =\psi ^i$, we obtain $g^{ab} D_a D_b \psi ^i=0$. And to satisfy $D_0 \Sigma^{0 a i}+ D_b \Sigma^{b a i}=0$, we use the previous solutions and let $D_0$ be determined in terms of them. And in order to satisfy the invariance under coordinates transformations on $M$, we let the equation $g^{ab} D_a D_b \psi ^i=0$ on $\Sigma_t$ come from pulling-back of the equation $g^{\mu\nu} D_\mu D_\nu \psi ^i=0$ on $M$, because the 3d space-like surface $\Sigma_t$ is immersed in the 4d space-time manifold $M$, thus pulling-back of 
\begin{equation}\label{eq:z73}
g^{\mu\nu} D_\mu D_\nu \psi ^i=0 , \quad \nabla_\mu g_{\nu \rho}=0,
\end{equation}
under the immersion map $i: \Sigma_t \to M$ gives $g^{ab} D_a D_b \psi ^i=0$, therefore $g_{\mu \nu}$ is also arbitrary metric on $M$. Regarding the discussion in the introduction, we let the metric in $g^{\mu\nu} D_\mu D_\nu \psi ^i=0$ be flat, at least locally, in order to avoid the influences of geometry of $M$ on the field $\psi ^i$ which relates with the spin current by the formula (\ref{eq:48}). 
\\

\begin{lemma}
By comparing the solutions $\Sigma^{abi}=v^{[a} D^{b]} \psi ^i$ and $ \Sigma ^{0 a i }=f D^a \xi ^i$ of the equations (\ref{eq:a20}) with the solutions (\ref{eq:a18}), and in order to get a correspondence between that solutions, e.g, by using a metric, we find that 
\begin{equation}\label{eq:a22}
\begin{split}
&\psi ^i=K^i b^i, \quad r_c^i=v_c b^i, \quad \xi ^i =K^i u^i , \\
and&\\
&Db^i=0, \quad Du^i=0, \quad dv=0 , \quad f=1,
\end{split} 
\end{equation}
for some vector fields $b^i, u^i  \in \Gamma\left(M; \mathfrak{so}(3, \mathbb{C})_P\right)$. 
\end{lemma}
By that we obtain the solutions
\begin{equation}\label{eq:a27}
\begin{split}
&E^{i}=\frac{1}{2}\epsilon^{abc} \Sigma_{ab}^i \partial_c=\frac{1}{2}\epsilon^{abc} v_a \left(\partial_b K ^i\right) b^i \partial_c  \in \Gamma\left(M;  T \Sigma \otimes \mathfrak{so}(3, \mathbb{C})_P\right), \\
and&\\
 &B^i  = \Sigma_{0a}^i dx^a=\left(d K^i \right)u^i\in \Gamma\left(M;  T^* \Sigma \otimes \mathfrak{so}(3, \mathbb{C})_P\right),
\end{split} 
\end{equation}
without needing using a specific metric. Thus the solutions are described by three complex scalar functions $K^i$, a Killing vector $v^a = g^{ab} v_b$ satisfying $v^a \partial_a K^i$=0, and two vector fields $b^i$ and $u^i$. In the vacuum, it must be $b^i=u^i$.
\\

\begin{remark}\label{eq:z2}
Regarding the solutions of the equations (\ref{eq:a18}) and (\ref{eq:a20}), we note that for every two solutions of $\Sigma_{ab}^i$ and $\Sigma^{abi}$, we get a metric $g_{ab}$ satisfying $\Sigma_{ab}^i=g_{aa'} g_{bb'} \Sigma^{a' b' i}$. Also we note that the metric used in $\Sigma_{ab}^i=g_{aa'} g_{bb'} \Sigma^{a' b' i}$ is not necessary the same metric used in $D^a \psi ^i=g^{ab} D_b \psi ^i$ for getting the solutions of (\ref{eq:a20}). Therefore the metric in BF theory is a derived quantity from the solutions of the 2-form $B  \in \Omega^2(M; \mathfrak{so}(3, \mathbb{C}))$.
\\
\end{remark}
\begin{remark}\label{eq:z3}
In the solution (\ref{eq:a27}), we see that $\Sigma^{IJ}$ can be written as $e^I \wedge e^J$, as required in constraint BF theory to get gravity theory, that is, according to self-dual projection, there is vector fields $b^I$ and $K^I$ satisfying $K^i b^i=P^i_{IJ} b^I (K^J b^J)$, therefore $\Sigma^{IJ}_{ab}=\left(b^I v_{[a} \right)D_{b]}\left( K^J b^J\right) $, then we can write $e^I_a=v_{a}b^I$ and $e^J_b=\left(\partial_{b} K^J\right)b^J$. And from $\Sigma_{0a}^{IJ}=u^I \left(\partial_a K^J \right)u^J$, we get $e^J_a=\left(\partial_{a} K^J\right)u^J$ and $e^I_0=v_0 u^I$, for $v_0=1$. A more general case is to find three vector fields $b_1^I$, $b_2^I$ and $K^I$ satisfying $K^i b^i=P^i_{IJ} b_1^I (K^J b_2^J)$, therefore $\Sigma^{IJ}_{ab}=\left(b_1^I v_{[a} \right)D_{b]}\left( K^J b_2^J\right) $, then we can write $e^I_a=v_{a}b_1^I$ and $e^J_b=\left(\partial_{b} K^J\right)b_2^J$. And from $\Sigma_{0a}^{IJ}=u_1^I \left(\partial_a K^J \right)u_2^J$, we get $e^J_a=\left(\partial_{a} K^J\right)u_2^J$ and $e^I_0=v_0 u_1^I$, for $v_0=1$. By that the (\ref{eq:a27}) can be written as $\Sigma^i=P^i_{IJ} \Sigma^{IJ}$ for $\Sigma^{IJ}=e^I \wedge e^J$. But we have to note that the solution (\ref{eq:a27}) is a general solution and we have to find a special solution, like to let $b^i$ be constant field and write $u^i$ in terms of it, as we will do in the following study.
\\
\end{remark}

Using the solution $ \Sigma ^{0a  i }=-\Sigma ^{ a0 i }= g^{ab} D_b \xi ^i=g^{ab} (\partial_b K ^i) u^i$ in the second equation of (\ref{eq:a20}), implies 
\begin{equation}\label{eq:a21}
\begin{split} 
J^{0 }_k=Q_k&=-2i\epsilon  _{kij} g^{ab} \left(b^i \partial_a K ^i\right)  \left(u^j \partial_b K^j \right) \\
&= - i\epsilon_{kij}g^{ab} \left( {b^i \partial _a K^i u^j \partial_b K^j  - b^j \partial _a K^j u^i \partial_b K^i } \right)\\
&=-i \epsilon  _{kij}g^{ab} \left( \partial_a K ^i\right)  \left(\partial_b K^j \right)(b^i u^j-b^j u^i).
\end{split} 
\end{equation}
We see that $J^{0 }_k \ne 0$ takes place only when $b \ne u$. Therefore in the vacuum it must be $b= u$. Regarding the equation (\ref{eq:z73}), the field $\psi ^i=K^i b^i$ satisfies $D^2 \psi ^i=0$, so $D b^i=0$ implies $\nabla^2 K^i=0$. If $J^{a }_k \ne 0$, we let $\Sigma^{abi}=v^{[a} (\partial^{b]} g ^i ) u^i$ with some field $g^i\ne K^i$. 
\\

If the charges $J^{0 i} \ne 0$ are given as a functions on $M$, and in order to get a solution using them, we let $b^i \in \Gamma\left(M; \mathfrak{so}(3, \mathbb{C})_P\right)$ be constant field on $M$, so we can determine the scalar functions $K ^i$ using the equation (\ref{eq:a21}), and so obtaining the vector $v\in \Gamma\left(M; T \Sigma\right)$ using $v^a \partial_a K^i=0$. But to satisfy $D b^i  = db^i  + \varepsilon ^i _{jk} A^j b^k  = 0$ for a constant vector field $b^i$, the connection $A^i_\mu$ must be written as $A^i_\mu=A_\mu b^i$. And we choose $u^i=b^i+f(x) b^i +a^i$, for a constant $a^i \in \Gamma\left(M; \mathfrak{so}(3, \mathbb{C})_P\right)$ satisfying $a_i b^i=0$, the function $f$ is needed for satisfying $Du^i=0$. We will see examples of determining $b^i$ and $u^i$ in spherical and cylindrical symmetries. 
\\

Another way of getting the solutions using the charges $J^{0 i} \ne 0$ is by using the solutions of $K^i$(obtained from $\nabla^2 K^i=0$) in the equation (\ref{eq:a21}), and by using $u^i=b^i+f(x) b^i +a^i$ (for constant $a^i$), one gets the field $b^i$, so getting the connection $A^i$ from $D b^i  =0$. We obtain $B^{ai}$ and $E^{ai}$ using the equations (\ref{eq:a27}), and obtain the matrix $\chi$ using the equation (\ref{eq:a26}), so obtaining the curvature $F=\chi \Sigma$. We note that $v^a E^{i}_a=0$, $v_a B^{ai}=0$ and $v^a \partial_a K^i=0$ depend on the symmetry of the system, for example, spherical symmetry, cylindrical symmetry, and so on. Thus we have seen that we can solve the equations of motion of BF theory without need using a specific metric on the manifold $M$, and the metric can be obtained from the solutions of $(\Sigma^{0ai}, \Sigma^{abi})$ and $(\Sigma^i_{0a}, \Sigma^i_{ab})$ according the remark (\ref{eq:z2}).
\\

\section{Solutions in spherical symmetric system}
We have seen that we can solve the equations of motion in BF theory by using a complex vector field $\psi^i=K^i b^i  \in \Gamma\left(M; \mathfrak{so}(3, \mathbb{C})_P\right)$ which allows us to obtain $v, E^i, B^i$ and $J^{0i}$ according to the equations (\ref{eq:a22}), (\ref{eq:a27}) and (\ref{eq:a21}). We try to find the solutions in spherical symmetric system in the vacuum($u=b$) and then apply it for matter located at a point. As we have seen that the solution of the system regards the symmetry of that system since we search for a vector $v\in \Gamma\left(M; T \Sigma\right)$ that satisfies $v^a D_a \psi ^i=0$, $v^a E^{i}_a=0$ and $v^a B^{i}_a=0$. For example, in spherical symmetry, we use the spherically coordinates $(r, \theta, \varphi)$ on the space-like slice $\Sigma_t=\Sigma=\mathbb{R}^3$. And according to the equation (\ref{eq:z41}), we have $D^2 \psi^i=0$ on $\Sigma_t$ without specifying the used metric, therefore we let it be the standard metric in the spherical coordinates. In spherical symmetry and static case, we let the vector field $\psi ^i$ depend only on the radius $r$, we get (for $D b^i=0$)
\begin{equation}\label{eq:1000}
D^2 \psi ^i  =D^2 (K^i b^i)=b^i\nabla^2 K^i=b^i \frac{1}{{r^2 }}\frac{\partial }{{\partial r}}\left( {r^2 \frac{\partial }{{\partial r}}K ^i } \right) = 0 \Rightarrow K ^i  =  \frac{{c^i }}{r},
\end{equation}
so $ \psi ^i  =c^i b^i/r$, for some constants $c^i\in \mathbb{R}$. Actually we can include $c^i$ in $b^i$ and just write $ \psi ^i  = b^i/r$. Therefore 
\[
D\psi ^i  = b^i (d K ^i  )= b^i \left( {d r\partial _r  + d \theta  \partial _\theta   + d \varphi \partial _\varphi  } \right)\frac{{1}}{r} =  - \frac{{ b^i }}{{r^2 }}d r,
\]
thus the 1-form $v$ ($dv=0$, $g^{ab} v_a D_b \psi ^i=0$) is
\[
v=a_1  d\theta    +a_2 d \varphi , \text{  } a_1, a_2 \in \mathbb{R},
\]
where in the spherical symmetry we let $a_1$ and $a_2$ do not depend on the coordinates $\theta$ and $\varphi$. The values of the constants $a_1$ and $a_2$ are not significant since $a^i=g^{ij} a_j$ is Killing vector, thus we set $a_1=a_2=1$. The used metric $g_{ab}$ here is the standard metric in the spherical coordinates, because we do not define any other metric. 
\\

Using the equation (\ref{eq:a27}), we get the solutions of the 1-form $B^i$ and the vector field $E^i$, 
\[
B^i  =(d K ^i  ) u^i =  - \frac{{b^i }}{{r^2 }}dr,  \quad for \quad u=b,
\]
\begin{equation*}
\begin{split}
E^i  = \frac{1}{2 }\varepsilon^{abc}\Sigma_{bc}^i \partial_a= \frac{1}{2 } \varepsilon^{abc} v_b  \left( \partial_c K^i \right) b^i\partial_a& = -\varepsilon ^{\varphi \theta r} \frac{{ b^i }}{{2r^2 }}\partial_\varphi  - \varepsilon ^{\theta \varphi r} \frac{{ b^i }}{{2r^2 }}\partial_\theta \\
&= \varepsilon ^{\varphi r\theta } \frac{{ b^i }}{{2r^2 }}\partial_\varphi  + \varepsilon ^{\theta r \varphi } \frac{{ b^i }}{{2r^2 }}\partial_\theta. 
\end{split} 
\end{equation*}
By that we get 
\begin{equation}\label{eq:a43}
\Sigma_{0r}^i =  - \frac{{b^i }}{{r^2 }}, \quad \Sigma_{r\theta}^i=\frac{{ b^i }}{{2r^2 }},  \quad \Sigma_{r\varphi}^i=\frac{{ b^i }}{{2r^2 }},
\end{equation}
while the other components like $\Sigma_{0 \theta}^i,\Sigma_{0\varphi}^i,...$, can be obtained by using gravitational fields $e^I_\mu$ derived from the solutions (\ref{eq:a43}). That is according to self-dual map, there are at least two constant fields $b_1^I$ and $b_2^I$ satisfying $b^i=P^i_{IJ} b_1^I b_2^J$. So from $\Sigma_{\mu \nu}^i=P^i_{IJ} e_\mu^I e_\nu^J$(regarding remark (\ref{eq:z3})), we get the gravitational fields 
\begin{equation*}\label{eq:z70}
e_0^I=- \frac{b_1^I }{r}, \quad e_r^I= \frac{b_2^I }{r} , \quad e_\theta^I=- \frac{b_1^I }{2r} ,  \quad e_\varphi^I=- \frac{b_1^I }{2r}.
\end{equation*}
We obtain the matrix $\chi$ using the equation (\ref{eq:a26}) with the solution (\ref{eq:1000}),
\begin{equation*}
\chi=\left( K_j  \delta^i_j\right)=\frac{1}{{r }} \text{diag}(c^1, c^2, c^3), 
\end{equation*}
where the constants $c^i$ have to be determined in order to satisfy the condition $tr\chi=0$(in the vacuum), so $\sum\limits_{i = 1}^3 {c_i  = 0} $. Thus we get the curvature $F=\chi \Sigma+\chi' \bar\Sigma$ (with setting $\chi'=0$ in the vacuum \cite{Kirill}),
\begin{equation}\label{eq:a40}
\begin{split}
& F_{0r}^i  = \chi ^i {_j} \Sigma _{0r}^j  =- \frac{{c^i b^i}}{{r^3 }} , \quad  \quad F_{r \theta}^i  = \chi ^i {_j} \Sigma_{r \theta}^j  =\frac{{c^i  b^i }}{{2r^3 }},\\
and &\\
& F_{r \varphi}^i  = \chi ^i {_j} \Sigma_{r \varphi}^j  =\frac{{ c^i b^i }}{{2r^3 }}.
\end{split} 
\end{equation}
Now we calculate the connection $A^i$ and the field $b^i$ which satisfies $D b^i=0$. Using $ F^i  = dA^i  + \varepsilon ^i {_{jk}} A^j  \wedge A^k $, we obtain
\begin{equation*}
\begin{split} 
& F_{0r}^i  = \frac{1}{2}\left( {\partial _0 A_r^i  - \partial _r A_0^i } \right) + \varepsilon ^i {_{jk}} A_0^j A_r^k , \\ 
& F_{\theta r}^i  = \frac{1}{2}\left( {\partial _\theta  A_r^i  - \partial _r A_\theta ^i } \right) + \varepsilon ^i {_{jk}} A_\theta ^j A_r^k , \\ 
& F_{\varphi r}^i  = \frac{1}{2}\left( {\partial _\varphi  A_r^i  - \partial _r A_\varphi ^i } \right) + \varepsilon ^i {_{jk}} A_\varphi ^j A_r^k. 
\end{split} 
\end{equation*}
Since we study a spherical symmetric system, we let the connection $A^i$ depend only on $r$. If we choose the gauge $A_r^i=0, i=1,2,3$, we get 
\[
F_{0r}^i  =  - \frac{1}{2}\partial _r A_0^i ,\quad \quad F_{\theta r}^i  =  - \frac{1}{2}\partial _r A_\theta ^i ,\quad \quad F_{\varphi r}^i  =  - \frac{1}{2}\partial _r A_\varphi ^i ,
\]
therefore by using the solution (\ref{eq:a40}), we obtain
\[
 - \frac{1}{2}\partial _r A_0^i  =  - \frac{{c^i b^i }}{{r^3 }},\quad \quad \frac{1}{2}\partial _r A_\theta ^i  = \frac{{c^i  b^i }}{{2r^3 }},\quad \quad \frac{1}{2}\partial _r A_\varphi ^i  = \frac{{ c^i b^i }}{{2r^3 }}.
\]
But $D b^i  = db^i  + \varepsilon ^i _{jk} A^j b^k  = 0$ and $\partial _\mu b^i=0$ for $\mu\ne r$, therefore
\begin{equation*}
\begin{split}
 &\partial _r b^i  + \varepsilon ^i {_{jk}} A_r^j b^k  = \partial _r b^i  = 0,\quad \quad \partial _0 b^i  + \varepsilon ^i {_{jk}} A_0^j b^k  = \varepsilon ^i {_{jk}} A_0^j b^k  = 0, \\ 
& \partial _\theta  b^i  + \varepsilon ^i {_{jk}} A_\theta ^j b^k  = \varepsilon ^i {_{jk}} A_\theta ^j b^k   = 0,\quad \quad \partial _\varphi  b^i  + \varepsilon ^i {_{jk}} A_\varphi ^j b^k  =\varepsilon ^i {_{jk}} A_\varphi ^j b^k = 0.
\end{split} 
\end{equation*}
Therefore the field $b^i \in \Gamma\left(M; \mathfrak{so}(3, \mathbb{C})_P\right)$ is constant and we get the solution
\[
A_0^i  =  - \frac{{c^i b^i }}{{r^2 }},\quad \quad A_\theta ^i  =  - \frac{{c^i  b^i }}{{2r^2 }},\quad \quad A_\varphi ^i  =  - \frac{{ c^i b^i }}{{2r^2 }},
\]
where we used $\varepsilon ^i {_{jk}} b ^j b^k  =0$. Thus in this solution the field $b^i \in \Gamma\left(M; \mathfrak{so}(3, \mathbb{C})_P\right)$ is constant on $M=\Sigma \times \mathbb{R}$. Next we try to find $b^i$ in the case of matter located at a point.
\\

\subsection{Solutions for matter located at a point}
If we have matter located at a point in $\Sigma_t=\Sigma=\mathbb{R}^3$, thus we have a spherical symmetry system in a static case $J^0_ i \ne 0$, $J^a_ i = 0$. We let that point be the origin $(0,0,0) \in \mathbb{R}^3$, therefore the charge (\ref{eq:a21}) is given by $Q^i(x)=Q^i_0 \delta^3(x)$, so $\int\limits_{\mathbb{R}^3} {Q^i_0 \delta ^3 (x)}  = Q^i_0=constant $(conservation of the charges). In order to get a same solution, as in the equations (\ref{eq:a43}) and (\ref{eq:a40}), we keep the field $b^i$ be constant, and in the formula (\ref{eq:a21})
\begin{equation}\label{eq:a42}
\begin{split} 
J^{0 }_k(x)=Q_k(x)=-i \epsilon  _{kij}\left( \partial_a K^i\right)  \left(\partial^a K^j \right)(b^i u^j-b^j u^i),
\end{split} 
\end{equation}
we use $Q^i(x)=Q^i_0 \delta^3(x)$. 
\\

In spherical symmetry, the functions $K^i$ are given by $K^i=c^i/r$, the equation (\ref{eq:1000}), therefore
\begin{equation*}
Q_k(x)=-2i g^{rr}\left( \partial_r K^i \partial_r K^j\right)   \left( \epsilon  _{kij} b^i u^j\right)=-i 2 \frac{c^i c^j}{r^4} \left( \epsilon  _{kij} b^i u^j\right).
\end{equation*}
Therefore in order to get $Q^i=Q^i_0 \delta^3(x)$, we replace $1/{r^4} $ with $ 1/(r^4+\epsilon^4)$, for some infinitesimal parameter $\epsilon\to 0^+ $, and we choose a solution for the field $u^j$ like
\[
u^j=b^j+\epsilon f b^j +\epsilon a^j /({-2i\pi^2  \sqrt{2} }),
\]
for some function $f$ on $M$ that is needed for satisfying $Du^i=0$ and a constant vector field $a^i \in \Gamma\left(M; \mathfrak{so}(3, \mathbb{C})_P\right)$. With that we obtain (for $i,j\ne k$)
\begin{equation*}
Q_k(r)=-2i  \frac{c^i c^j}{r^4+\epsilon^4} \epsilon  _{kij} b^i \left( b^j+\epsilon f b^j +\frac{\epsilon}{-2i\pi^2  \sqrt{2}} a^j \right)= \frac{1}{\pi^2\sqrt{2}}\frac{\epsilon}{r^4+\epsilon^4} \epsilon  _{kij} (c^i b^i) (c^j a^j).
\end{equation*}
\\
Comparing with $Q^i(x)=Q^i_0 \delta^3(x)$, we find $ \epsilon ^k {_{ij}} (c^i b^i) (c^j a^j)=Q_0^k=constant$, and by imposing $(c_i a_i) (c^i a^i)=1$ with $(c_i b_i) (c^i a^i)=0$, we find 
\[
(cb)^2=(c_i b_i) (c^i b^i)=Q_{0k} Q_0^k, 
\]
thus we choose $c^i b^i=e^i \sqrt{Q_{0k} Q_0^k} $ for $\left\| e^i \right\|=1$, so the constant field $b^i$ is determined by $ \sqrt{Q_{0k} Q_0^k} $ with free ${SO}(3, \mathbb{C})$ rotation.
\\

By that, for $r>0$, $\epsilon\to 0^+$, we obtain the same solutions as in the equations (\ref{eq:a43}) and (\ref{eq:a40}), but with $c^i b^i=e^i \sqrt{Q_{0k} Q_0^k} $ for $\left\| e^i \right\|=1$ and $\sum\limits_{i = 1}^3 {c_i  = 0} $. Since $c^i b^i$ is finite value, it is not sufficient to let the constants $c_i $ take arbitrary values, so we choose them to be $(c_i)=(1, 1, -2)$.
\\

By that we have given an example for the possibility of solving the equations of motion in BF theory without need using a gravitational metric on $M$, we just need using a vector field $\psi^i \in \Gamma\left(M; \mathfrak{so}(3, \mathbb{C})_P\right)$ which is defined in the spin current $J^\mu_ k = 2i \varepsilon  _{kij} \left( D_\nu \psi ^i\right)  \Sigma ^{\mu \nu j }$ of matter, the equation (\ref{eq:48}). Also we saw that the solutions depend on the symmetry of the system, since we need for obtaining the solutions some vector $v$ that satisfies $v^a \partial_a K^i=0$, $v^a E^{i}_a=0$, $v_a B^{ai}=0$.

\section{Solutions in cylindrical symmetric system}
In cylindrical symmetric system, we let matter be homogeneously located along the $Z$-axis. As we did in spherically symmetry, we search for the field $\psi^i=K^i b^i  \in \Gamma\left(M; \mathfrak{so}(3, \mathbb{C})_P\right)$, then we calculate $v, E^i, B^i$ and $J^{0i}$ according to the equations (\ref{eq:a22}), (\ref{eq:a27}) and (\ref{eq:a21}). The vector $v\in \Gamma\left(M; T \Sigma\right)$ satisfies $v^a D_a \psi ^i=0$, $v^a E^{i}_a=0$ and $v^a B^{i}_a=0$, thus it is Killing vector. We try to find the solution in the vacuum($u=b$) and then apply it for matter located homogeneously along the $Z$-axis. The needed information for solving the equations of motion is only the spin charge $Q^i(x)$, the equation (\ref{eq:a21}). As we mentioned before, we do not need to use a gravitational metric, we just use a standard metric. In cylindrical symmetry, we use the cylindrical coordinates $( \rho, \varphi, z)$ on the space-like slice $\Sigma_t=\Sigma=\mathbb{R}^3$ of constant time $t$. We let the vector field $\psi ^i$ depend only on the radius $\rho$, we get (for $D^2 \psi ^i  =0$ and $D b^i=0$)
\begin{equation}\label{eq:1000s}
D^2 \psi ^i  =D^2 (K^i b^i)=b^i\nabla^2 K^i=b^i \frac{1}{{\rho }}\frac{\partial }{{\partial \rho}}\left( {\rho \frac{\partial }{{\partial \rho}}K ^i } \right) = 0 \Rightarrow K ^i  =  {c^i } \log(\rho),
\end{equation}
so $ \psi ^i  =c^i b^i \log(\rho)$, for some constants $c^i\in \mathbb{R}$. Therefore 
\[
D\psi ^i  =c^i b^i (d K ^i  )= c^i b^i \left( {d \rho \partial _\rho    + d \varphi \partial _\varphi  + d z \partial _z} \right)\log(\rho) =   \frac{{ c^i b^i }}{{\rho }}d \rho,
\]
thus the 1-form $v$ ($dv=0$, $g^{ab} v_a D_b \psi ^i=0$) is
\[
v=a_1  d\varphi    +a_2 d z , \text{  } a_1, a_2 \in \mathbb{R},
\]
where in the cylindrical symmetry we let $a_1$ and $a_2$ do not depend on the coordinates $z$ and $\varphi$. The values of the constants $a_1$ and $a_2$ are not significant since $a^i=g^{ij} a_j$ is Killing vector, thus we set $a_1=a_2=1$. The used metric $g_{ab}$ here is the standard metric in the cylindrical coordinates, because we do not define any other metric. 
\\

Using the equation (\ref{eq:a27}), we get the solutions of the 1-form $B^i$ and the vector field $E^i$, 
\[
B^i  = \frac{{ c^i u^i }}{{\rho }}d \rho=   \frac{{ c^i b^i }}{{\rho }}d \rho,  \quad for \quad u=b,
\]
\begin{equation*}
E^i  = \frac{1}{2 }\varepsilon^{abc}\Sigma_{bc}^i \partial_a= \frac{1}{2 } \varepsilon^{abc} v_b  \left( \partial_c K^i \right) b^i\partial_a = \varepsilon ^{\varphi z \rho} \frac{{ c^i b^i }}{{2\rho }}\partial_\varphi  + \varepsilon ^{z \varphi \rho} \frac{{c^i b^i }}{{2\rho }}\partial_z .
\end{equation*}
By that we get 
\begin{equation}\label{eq:a43s}
\Sigma_{0\rho}^i = -\Sigma_{\rho 0}^i =  \frac{{c^i b^i }}{{\rho }}, \quad \Sigma_{\rho z}^i=-\Sigma_{z \rho}^i=-\frac{{c^i b^i }}{{2\rho }},  \quad \Sigma_{\rho\varphi}^i=-\Sigma_{\varphi \rho}^i=- \frac{{c^i b^i }}{{2\rho }},
\end{equation}
while the other components like $\Sigma_{0 z}^i,\Sigma_{0\varphi}^i,...$ are zeros.
\\

We obtain the matrix $\chi$ using the equation (\ref{eq:a26}) with the solution (\ref{eq:1000s}),
\begin{equation*}
\chi=\left( K_j  \delta^i_j\right)= \log(\rho) \text{diag}(c^1, c^2, c^3), 
\end{equation*}
where the constants $c^i$ have to be determined in order to satisfy the condition $tr\chi=0$(in the vacuum), so $\sum\limits_{i = 1}^3 {c_i  = 0} $. Thus we get the curvature $F=\chi \Sigma+\chi' \bar\Sigma$ (with setting $\chi'=0$ in the vacuum \cite{Kirill}),
\begin{equation}\label{eq:a40s}
\begin{split}
& F_{0\rho}^i  = \chi ^i {_j} \Sigma _{0\rho}^j  = \frac{{\log(\rho)}}{{\rho }}c^i c^i b^i , \quad  \quad F_{\rho z}^i  = \chi ^i {_j} \Sigma_{\rho z}^j  =-\frac{\log(\rho)}{{2\rho }}c^i c^i b^i,\\
and &\\
& F_{\rho \varphi}^i  = \chi ^i {_j} \Sigma_{\rho \varphi}^j  =-\frac{{ \log(\rho)}}{{2\rho }}c^i c^i b^i.
\end{split} 
\end{equation}
Using the gauge $A_\rho^i=0, i=1,2,3$, with letting $A^i$ depend only on $\rho$, we obtain
\[
A_0^i  = - (\log(\rho))^2( c^i)^2 b^i,\quad \quad A_z ^i  =-\frac{1}{2}  (\log(\rho))^2( c^i)^2 b^i,\quad \quad A_\varphi ^i  =-\frac{1}{2} (\log(\rho))^2( c^i)^2 b^i,
\]
where the field $b^i \in \Gamma\left(M; \mathfrak{so}(3, \mathbb{C})_P\right)$ is constant on $M=\Sigma \times \mathbb{R}$. 
\\

As we did in spherical symmetry, we find $b^i$ by using the spin charge $Q^i(x)$ which is given by the equation (\ref{eq:a21}). Since the system is static and the matter homogeneously located along the $Z$-axis, the spin charge $Q^i(\rho, \varphi, z)$ is given by $Q^i(\rho, \varphi, z)=Q^i_0 \delta(\rho)/2 \pi\rho$, which yields $\int\limits_0^\infty \int\limits_0^{2 \pi} \rho d\varphi d\rho  Q^i_0 \delta(\rho)/2 \pi\rho=Q^i_0$ for each point of $Z$. Here $Q^i_0$ is point charge located at each point of $Z$-axis. 
\\
\\
In order to get a same solution, as in the equations (\ref{eq:a43s}) and (\ref{eq:a40s}), we keep the field $b^i$ be constant, and in the formula (\ref{eq:a21})
\begin{equation}\label{eq:a42s}
\begin{split} 
J^{0 }_k(x)=Q_k(x)=-i \epsilon  _{kij}\left( \partial_a K^i\right)  \left(\partial^a K^j \right)(b^i u^j-b^j u^i),
\end{split} 
\end{equation}
we use $Q^i(x)=Q^i_0 \delta(\rho)/2 \pi \rho$.
\\
\\
In cylindrical symmetry, the functions $K^i$ are given by $K^i=c^i \log(\rho)$, the equation (\ref{eq:1000s}), therefore
\begin{equation*}
Q_k(x)=-2i g^{\rho \rho}\left( \partial_\rho K^i \partial_\rho K^j\right)   \left( \epsilon  _{kij} b^i u^j\right)=- 2i \frac{c^i c^j}{\rho^2} \left( \epsilon  _{kij} b^i u^j\right).
\end{equation*}
Therefore in order to get $Q^i=Q^i_0 \delta(\rho)/2 \pi\rho$, we replace $1/{\rho^2} $ with $ 1/{\rho \rho^{1-\epsilon}}$, for some infinitesimal parameter $\epsilon\to 0^+ $, and we choose a solution for the field $u^j$ like
\[
u^j=b^j+\epsilon f b^j +\epsilon a^j /({-4i\pi}),
\]
for some scalar function $f$ on $M$ that is needed for satisfying $Du^i=0$, with a constant vector field $a^i \in \Gamma\left(M; \mathfrak{so}(3, \mathbb{C})_P\right)$. By that we obtain (for $i,j\ne k$)
\begin{equation*}
Q_k(r)=-2i  \frac{c^i c^j}{{\rho \rho^{1-\epsilon}}} \epsilon  _{kij} b^i \left( b^j+\epsilon f b^j +\frac{\epsilon}{{-4i\pi}} a^j \right)= \frac{1}{2\pi \rho}\frac{\epsilon}{ \rho^{1-\epsilon}} \epsilon  _{kij} (c^i b^i) (c^j a^j).
\end{equation*}
\\
Comparing with $Q^i=Q^i_0 \delta(\rho)/2 \pi\rho$, we find $ \epsilon ^k {_{ij}} (c^i b^i) (c^j a^j)=Q_0^k=constant$, and by imposing $(c_i a_i) (c^i a^i)=1$ with $(c_i b_i) (c^i a^i)=0$, we find 
\[
(cb)^2=(c_i b_i) (c^i b^i)=Q_{0k} Q_0^k, 
\]
thus we choose $c^i b^i=e^i \sqrt{Q_{0k} Q_0^k} $ for $\left\| e^i \right\|=1$, so the constant field $c^i b^i$ is determined by $ \sqrt{Q_{0k} Q_0^k} $ with free ${SO}(3, \mathbb{C})$ rotation.
\\

By that, for $\rho>0$, $\epsilon\to 0^+$, we obtain the same solutions as in the equations (\ref{eq:a43s}) and (\ref{eq:a40s}), but with $c^i b^i=e^i \sqrt{Q_{0k} Q_0^k} $ for $\left\| e^i \right\|=1$ and $\sum\limits_{i = 1}^3 {c_i  = 0} $. Since $c^i b^i$ is finite value, it is not sufficient to let the constants $c_i $ take arbitrary values, so we choose them to be $(c_i)=(1, 1, -2)$.
\\

\section{Conclusions}
We have studied the BF theory including matter by redefinition the 2-form $B^i$, as $B^i+\varepsilon^i{_{jk}}\psi^j B^k$, or redefinition the Lagrangian multipliers $\varphi _{ij}$ as $ \varphi _{ij}  + \varepsilon _{\ell ik} \psi ^\ell  \chi^k {_j}$, so that we can get $DB^i=0$, in the case of non-zero spin current of matter fields. The new field $\psi^i$ is defined using the spin current vector $J^{\nu i}= (D_\mu \psi_k) \varepsilon {^{ki}}_j \Sigma^{\mu\nu j}$. We saw that we can solve the BF equations by using only the spin current of matter, that is it is enough to solve the equations $\delta S/ \delta A^i=0$, $D B^i=0$ and $J^{\nu i}= (D_\mu \psi_k) \varepsilon {^{ki}}_j \Sigma^{\mu\nu j}$ without using a gravitational metric on $M$ and without needing solving the equation $\delta S/ \delta B^i=0$ which includes the Lagrangian multiplier $\varphi_{ij}$(a non-physical variable), so we get $\varphi_{ij}$ by using the solutions in $\delta S/ \delta B^i=0$. We found that to obtain the solutions of BF theory, it is enough to use(find) the field $\psi^i$ and the Killing vector $v$(satisfies $v^a D_a\psi^i=0$) in euclidean coordinates, where it is convenient to describe the spin currents and their lines in euclidean coordinates and not need to describe them in curved coordinates. Also it is possible to obtain the solutions of BF theory using only the charges $J^{0 i} \ne 0$ when they are given as a functions on $M$ in the static case (discussion below the equation (\ref{eq:a21})). We saw that the singularities appear in solution of $\psi^i$, that relates to the idea that the spin current $J^{\mu i}$ is the source for $\psi^i$, therefore $\psi^i$ has singularities on line of that spin current, and by that the singularities appear and not by using a gravitational metric. We found that the solutions of BF theory equations depend on the symmetry of the system and every two solutions of $(\Sigma^i_{ab}, \Sigma^i_{0a})$ and $(\Sigma^{abi}, \Sigma^{0ai})$ determine a metric(remark \ref{eq:z2}), so there is no specific metric needed in BF theory because it is a topological theory, and those solutions are able to be written as $e^I\wedge e^J$ (remark \ref{eq:z3}). Finally we applied the solutions of BF theory in a spherical and cylindrical symmetric systems in static case of matter.

\section{Acknowledgements}
I am grateful to Professor Yoshihiro Fukumoto at Ritsumeikan University for many useful notes on the manuscript.

\section{Appendix A}
$\rm{I}$- We test satisfying of the equation of motion of classical general relativity in the vacuum using the solutions in BF theory, the equations have not to include the  non-physical variable $\varphi$. The canonical formalism of GR in the vacuum gives the constraints \cite{Carlo}
\begin{equation}\label{eq:ab3}
D_a E_i^a  = 0, \quad E_i^a F_{ab}^i  = 0,\quad C=\varepsilon _{ijk} E^{ai} E^{bj} F_{ab}^k  = 0.
\end{equation}
These constraints are generators of gauge symmetry and diffeomorphism invariance up to boundary terms(ignoring the boundary terms). We get $D_a E^{ai}=0$ from $\varepsilon ^{\mu \nu \rho \sigma } D_\nu  B_{\rho \sigma }^i =0$ by setting $\mu=0$ to get $\varepsilon ^{0abc } D_a  B_{b c}^i =0$, and using $\varepsilon ^{0abc }\equiv \varepsilon ^{abc }$ to get $\varepsilon ^{abc } D_a  B_{b c}^i =2 D_a E^{ai} =0$, where $E^{ai}=\varepsilon ^{abc } B_{b c}^i/2$ is conjugate to the connection $A_a^i$ on space-like slice of constant time on which we use the coordinates $(x^a)_{a=1,2,3}$. Thus the constraint $D_a E^{ai}=0$ is satisfied in BF theory.
\\

Using the equation (\ref{eq:50}), in the vacuum ${\chi'}^i {_j}=0$, it reduces to $F^i (A)  =\chi ^i {_j } \Sigma^j$, on the spacelike surface of constant time $\Sigma_t$, it becomes $F^i (A)_{ab}  =\chi ^i {_j } \Sigma^j_{ab}$. Using $E^{ai}=\varepsilon^{abc} \Sigma^i_{bc} $, we obtain
\[
F^i (A)_{ab}  =\chi ^i {_j } \Sigma^j_{ab}=\chi ^i {_j } \varepsilon_{abc} E^{j c}.
\]
Multiplying it by $E_i^a$ and summing over the contracted indices, we get
\[
E_i^a F^i (A)_{ab}  =E_i^a \chi ^i {_j } \varepsilon_{abc} E^{j c}= \chi_{ij } \varepsilon_{abc}E^{ai} E^{j c}=0,
\]
where we used the fact that the matrix $\chi_{ij }$ is symmetric. Therefore the second constraint of (\ref{eq:ab3}) is satisfied. Using $F^i (A)_{ab}  =\chi ^i {_j } \varepsilon_{abc} E^{j c}$ in $C$, yields
\[
C=\varepsilon _{ijk} E^{ai} E^{bj} \chi ^k {_{i'} } \varepsilon_{abc} E^{i' c} = \chi {_k}^{ i'}  \varepsilon^{ijk} \varepsilon_{abc}  E^{a}_i E^{b}_j E^c_{i'} .
\]
Then using $E^{a}_i=e e^{a}_i$, where $e^{a}_i$ is the inverse of the gravitational field $e_{a}^i$ and $e=\det(e_{a}^i)$, we obtain
\[
C=e^3 \chi {_k}^{ i'}  \varepsilon^{ijk} \varepsilon_{abc}  e^{a}_i e^{b}_j e^c_{i'}=e^3 \chi {_k}^{ i'}  \varepsilon^{ijk} e^{-1}  \varepsilon_{ij i'}=2 e^2 tr \chi .
\]
But in the vacuum $ tr \chi =0$(the equation (\ref{eq:55})), thus $C=\varepsilon _{ijk} E^{ai} E^{bj} F_{ab}^k  = 0$ is satisfied. By that we find that the general relativity constraints are satisfied in the vacuum using the equations of motion of BF theory. As required, the general relativity equations do not include the matrices $\chi$ and $\varphi$ and do not include the new field $\psi^i$.
\\
\\
$\rm{II}$- A principal G-bundle consists of the following data:\\
1- a manifold $P$, called the total space,\\
2- a Lie group $G$ acting freely on $P$ on the right:
\[
P\times G \to P, \quad (p,g)\mapsto pg.
\]
The free action means that the stabilizer of every point is trivial, that every element of $G$ (except the identity) moves every point in P. We assume that the space of orbits $M=P/G$ is a manifold (the base space). With projection $\pi: P\to M$ and for every $p\in M$, the submanifold $\pi^{-1}(p)\subset P$ is fibre over $M$. Let $\left\{ {U_\alpha  } \right\}$ be open cover of $M$, the local trivialization is G-equivariant diffeomorphisms 
\[
\psi _\alpha  :\pi ^{ - 1} (U_\alpha  ) \to U_\alpha   \times G,
\]
given by $\psi _\alpha  (p) = \left( {\pi (p),g_\alpha  (p)} \right)$ for some G-equivariant map $g_\alpha  :\pi ^{ - 1} (U_\alpha  ) \to G$. Equivariance means that $g_\alpha(pg)=g_\alpha(p) g$. We say that the bundle is trivial if there exists a diffeomorphism $\psi :P \to M \times G$ such that $\psi  (p) = \left( {\pi (p),\psi  (p)} \right)$ and such that $\psi(pg ) =\psi(p)g$. This last condition is simply the G-equivariance of $\psi$.

We separate $T_p P$ to vertical and horizontal vector spaces at each point $p\in P$, we get the vertical vector fields by acting of group $G$ on $P$ by
\[
\sigma _p (X) = \frac{d}{{dt}}\left. {\left( {pe^{tX} } \right)} \right|_{t = 0} ,
\]
for every vector $X\in \mathfrak{g}$ (where $ \mathfrak{g}$ is Lie algebra of $G$), this satisfies
\[
\pi _* \sigma _p (X) = \frac{d}{{dt}}\left. {\left( {\pi  \left( {pe^{tX} } \right)} \right)} \right|_{t = 0}  = \frac{d}{{dt}}\left. {\left( {\pi  \left( p \right)} \right)} \right|_{t = 0}  = 0,
\]
thus $\sigma _p (X) $ is vertical vector field at $p\in P$. In this bundle the connection is defined as a map
\[
D:\Gamma (TP) \to \Gamma (T_e G), \quad T_e G=lie(G),
\]
as usual, the connection is a map from the tangent space to itself. And since $P$ is locally product $U \times G$, then $\left. TP \right|_U =TU \times T_eG $, so locally
\[
D:\Gamma (TU) \times \Gamma(T_e G) \to \Gamma (T_e G).
\]
We can get this map by letting $D=\nabla+A$, for $A \in \Gamma (T^*U \otimes T_e G) $, and $\nabla$ is connection on $TU$.

\section{Appendix B: Calculating $ \frac{{\delta S_m }}{{\delta \Sigma^i  }} $}
Starting from
\begin{equation}\label{eq:60}
\left. {\frac{{\delta S_m }}{{\delta B_{\sigma _1 \rho _1 }^\ell  }}} \right|_{constriant}  = \frac{{\delta S_m }}{{\delta \Sigma _{\sigma _1 \rho _1 }^\ell  }} = \frac{{ - 2}}{{\sqrt { - g} }}\frac{{\delta S_m }}{{\delta g_{\mu _1 \nu _1 } }}\frac{{\sqrt { - g} }}{{ - 2}}\frac{{\delta g_{\mu _1 \nu _1 } }}{{\delta \Sigma _{\sigma _1 \rho _1 }^\ell  }} = T^{\mu _1 \nu _1 } \frac{{\sqrt { - g} }}{{ - 2}}\frac{{\delta g_{\mu _1 \nu _1 } }}{{\delta\Sigma _{\sigma _1 \rho _1 }^\ell  }},
\end{equation}
where $T^{\mu  \nu }$ is energy-momentum tensor, and $S_m$ is matter action without specifying. Using
\[
\delta g_{\mu _1 \nu _1 }  = \frac{1}{{\sqrt { - g} }}\delta \left( {\sqrt { - g} g_{\mu _1 \nu _1 } } \right) - \frac{1}{{\sqrt { - g} }}g_{\mu _1 \nu _1 } \delta \left( {\sqrt { - g} } \right),
\]
we obtain
\begin{equation}\label{eq:52}
\sqrt { - g} \frac{{\delta g_{\mu _1 \nu _1 } }}{{\delta \Sigma _{\sigma _1 \rho _1 }^\ell  }} = \frac{{\delta \left( {\sqrt { - g} g_{\mu _1 \nu _1 } } \right)}}{{\delta \Sigma _{\sigma _1 \rho _1 }^\ell  }} - g_{\mu _1 \nu _1 } \frac{{\delta \left( {\sqrt { - g} } \right)}}{{\delta \Sigma _{\sigma _1 \rho _1 }^\ell  }}=I_1 - I_2 .
\end{equation}
We use the Urbantke formula \cite{Ingemar}
\begin{equation}\label{eq:ab2}
\sqrt { - g} g_{\mu _1 \nu _1 }  = \varepsilon ^{\mu \nu \rho \sigma } \varepsilon _{ijk} \Sigma _{\mu _1 \mu }^i \Sigma _{\nu \rho }^j \Sigma _{\sigma \nu _1 }^k ,
\end{equation}
from which we get
\begin{equation*}
\begin{split}
& \delta \left( {\sqrt { - g} g_{\mu _1 \nu _1 } } \right) = \delta \left( {\varepsilon ^{\mu \nu \rho \sigma } \varepsilon _{ijk} \Sigma _{\mu _1 \mu }^i \Sigma _{\nu \rho }^j \Sigma _{\sigma \nu _1 }^k } \right) \\
&= \varepsilon ^{\mu \nu \rho \sigma } \varepsilon _{ijk} \left( {\delta \Sigma _{\mu _1 \mu }^i } \right)\Sigma _{\nu \rho }^j \Sigma _{\sigma \nu _1 }^k  + \varepsilon ^{\mu \nu \rho \sigma } \varepsilon _{ijk} \Sigma _{\mu _1 \mu }^i \left( {\delta \Sigma _{\nu \rho }^j } \right)\Sigma _{\sigma \nu _1 }^k  + \varepsilon ^{\mu \nu \rho \sigma } \varepsilon _{ijk} \Sigma _{\mu _1 \mu }^i \Sigma _{\nu \rho }^j \left( {\delta \Sigma _{\sigma \nu _1 }^k } \right).
\end{split}
\end{equation*}
\\
We use it for calculating $I_1$ in (\ref{eq:52}),
\begin{equation*}
\begin{split}
  I_1 & = \frac{\delta }{{\delta \Sigma _{\sigma _1 \rho _1 }^\ell  }}\left( {\sqrt { - g} g_{\mu _1 \nu _1 } } \right) \\
&= \varepsilon ^{\mu \nu \rho \sigma } \varepsilon _{ijk} \Sigma _{\nu \rho }^j \Sigma _{\sigma \nu _1 }^k \left( {\delta _\ell ^i \delta _{\mu _1 }^{\sigma _1 } \delta _\mu ^{\rho _1 } } \right) + \varepsilon ^{\mu \nu \rho \sigma } \varepsilon _{ijk} \Sigma _{\mu _1 \mu }^i \Sigma _{\sigma \nu _1 }^k \left( {\delta _\ell ^j \delta _\nu ^{\sigma _1 } \delta _\rho ^{\rho _1 } } \right) \\
 &\quad\quad \quad+ \varepsilon ^{\mu \nu \rho \sigma } \varepsilon _{ijk} \Sigma _{\mu _1 \mu }^i \Sigma _{\nu \rho }^j \left( {\delta _\ell ^k \delta _\sigma ^{\sigma _1 } \delta _{\nu _1 }^{\rho _1 } } \right)  \\
&= \varepsilon ^{\rho _1 \nu \rho \sigma } \varepsilon _{\ell jk} \Sigma _{\nu \rho }^j \Sigma _{\sigma \nu _1 }^k \left( {\delta _{\mu _1 }^{\sigma _1 } } \right) + \varepsilon ^{\mu \sigma _1 \rho _1 \sigma } \varepsilon _{i\ell k} \Sigma _{\mu _1 \mu }^i \Sigma _{\sigma \nu _1 }^k  + \varepsilon ^{\mu \nu \rho \sigma _1 } \varepsilon _{ij\ell } \Sigma _{\mu _1 \mu }^i \Sigma _{\nu \rho }^j \left( {\delta _{\nu _1 }^{\rho _1 } } \right) \\ 
 &= \varepsilon ^{\rho _1 \nu \rho \sigma } \varepsilon _{\ell ij} \Sigma _{\nu \rho }^i \Sigma _{\sigma \nu _1 }^j \left( {\delta _{\mu _1 }^{\sigma _1 } } \right) + \varepsilon ^{\mu \sigma _1 \rho _1 \sigma } \varepsilon _{i\ell j} \Sigma _{\mu _1 \mu }^i \Sigma _{\sigma \nu _1 }^j  + \varepsilon ^{\mu \nu \rho \sigma _1 } \varepsilon _{ij\ell } \Sigma _{\mu _1 \mu }^i \Sigma _{\nu \rho }^j \left( {\delta _{\nu _1 }^{\rho _1 } } \right) \\ 
& = \varepsilon ^{\rho _1 \nu \rho \sigma } \varepsilon _{\ell ij} \Sigma _{\nu \rho }^i \Sigma _{\sigma \nu _1 }^j \left( {\delta _{\mu _1 }^{\sigma _1 } } \right) - \varepsilon ^{\mu \sigma _1 \rho _1 \sigma } \varepsilon _{\ell ij} \Sigma _{\mu _1 \mu }^i \Sigma _{\sigma \nu _1 }^j  + \varepsilon ^{\mu \nu \rho \sigma _1 } \varepsilon _{\ell ij} \Sigma _{\mu _1 \mu }^i \Sigma _{\nu \rho }^j \left( {\delta _{\nu _1 }^{\rho _1 } } \right). \\ 
\end{split}
\end{equation*}
\\
To calculate $I_2$, we use
\[
\sqrt { - g}  = \frac{i}{6} \varepsilon ^{\mu \nu \rho \sigma } \delta _{ij} \Sigma _{\mu \nu }^i \Sigma _{\rho \sigma }^j ,
\]
hence
\begin{equation*}
\begin{split}
  I_2  &= g_{\mu _1 \nu _1 } \frac{\delta }{{\delta \Sigma _{\sigma _1 \rho _1 }^\ell  }}\left( {\sqrt { - g} } \right) =  \frac{i}{6} g_{\mu _1 \nu _1 } \varepsilon ^{\mu \nu \rho \sigma } \delta _{ij} \left( {\delta _\ell ^i \delta _\mu ^{\sigma _1 } \delta _\nu ^{\rho _1 } } \right)\Sigma _{\rho \sigma }^j + \frac{i}{6} g_{\mu _1 \nu _1 } \varepsilon ^{\mu \nu \rho \sigma } \delta _{ij} \Sigma _{\mu \nu }^i \left( {\delta _\ell ^j \delta _\rho ^{\sigma _1 } \delta _\sigma ^{\rho _1 } } \right)\\
 &=  \frac{i}{6} g_{\mu _1 \nu _1 } \varepsilon ^{\sigma _1 \rho _1 \rho \sigma } \delta _{i\ell } \Sigma _{\rho \sigma }^i  +\frac{i}{6} g_{\mu _1 \nu _1 } \varepsilon ^{\mu \nu \sigma _1 \rho _1 } \delta _{i\ell } \Sigma _{\mu \nu }^i  =  \frac{i}{3} g_{\mu _1 \nu _1 }  {\varepsilon ^{\mu \nu \sigma _1 \rho _1 } \delta _{i\ell } \Sigma _{\mu \nu }^i } .
\end{split}
\end{equation*}
Therefore
\begin{equation}\label{eq:53}
\begin{split}
 &\sqrt { - g} \frac{{\delta g_{\mu _1 \nu _1 } }}{{\delta \Sigma _{\sigma _1 \rho _1 }^\ell  }} = I_1  - I_2  \\ 
 & = \varepsilon ^{\rho _1 \nu \rho \sigma } \varepsilon _{\ell ij} \Sigma _{\nu \rho }^i \Sigma _{\sigma \nu _1 }^j \left( {\delta _{\mu _1 }^{\sigma _1 } } \right) - \varepsilon ^{\mu \sigma _1 \rho _1 \sigma } \varepsilon _{\ell ij} \Sigma _{\mu _1 \mu }^i \Sigma _{\sigma \nu _1 }^j  + \varepsilon ^{\mu \nu \rho \sigma _1 } \varepsilon _{\ell ij} \Sigma _{\mu _1 \mu }^i \Sigma _{\nu \rho }^j \left( {\delta _{\nu _1 }^{\rho _1 } } \right) \\ 
&  - \frac{i}{3}g_{\mu _1 \nu _1 } \left( {\varepsilon ^{\mu \nu \sigma _1 \rho _1 } \delta _{i\ell } \Sigma _{\mu \nu }^i } \right). 
\end{split}
\end{equation}
\\
By this the equation (\ref{eq:60}) becomes
\begin{equation}\label{eq:61}
\begin{split}
& \left. {\frac{{\delta S_m }}{{\delta B_{\sigma _1 \rho _1 }^\ell  }}} \right|_{constriant}  = \frac{{ - 1}}{2}T^{\mu _1 \nu _1 } \left( {I_1  - I_2 } \right) \\ 
 & = \frac{1}{2}T^{\mu _1 \nu _1 } \varepsilon ^{\mu \sigma _1 \rho _1 \sigma } \varepsilon _{\ell ij} \Sigma _{\mu _1 \mu }^i \Sigma _{\sigma \nu _1 }^j   - \frac{1}{2}T^{\mu _1 \nu _1 } \varepsilon ^{\mu \nu \rho \sigma _1 } \varepsilon _{\ell ij} \Sigma _{\mu _1 \mu }^i \Sigma _{\nu \rho }^j \left( {\delta _{\nu _1 }^{\rho _1 } } \right)   \\ 
&- \frac{1}{2}T^{\mu _1 \nu _1 } \varepsilon ^{\rho _1 \nu \rho \sigma } \varepsilon _{\ell ij} \Sigma _{\nu \rho }^i \Sigma _{\sigma \nu _1 }^j \left( {\delta _{\mu _1 }^{\sigma _1 } } \right)  + \frac{i}{6}T^{\mu _1 \nu _1 } g_{\mu _1 \nu _1 } {\varepsilon ^{\mu \nu \sigma _1 \rho _1 } \delta _{i\ell } \Sigma _{\mu \nu }^i } \\ 
&=I_1 +I_2 +I_3+I_4.
\end{split}
\end{equation}
Using $T^{\mu _1 \nu _1}=T^{IJ} e_I^{\mu _1} e_J^{\nu_1} $, the first term of (\ref{eq:61}) becomes
\begin{equation*}
\begin{split}
 2I_1&=T^{\mu _1 \nu _1 } \varepsilon ^{\mu \nu \sigma _1 \rho _1 } \varepsilon _{ij\ell } \Sigma _{\mu _1 \mu }^i \Sigma _{\nu \nu _1 }^j  = T^{IJ} e_I^{\mu _1 } e_J^{\nu _1 } \varepsilon ^{\mu \nu \sigma _1 \rho _1 } \varepsilon _{ij\ell } P_{KL}^i P_{K_1 L_1 }^j e_{\mu _1 }^K e_\mu ^L e_\nu ^{K_1 } e_{\nu _1 }^{L_1 }  \\ 
  &= T^{IJ} \delta _I^K \delta _J^{L_1 } \varepsilon ^{\mu \nu \sigma _1 \rho _1 } \varepsilon _{ij\ell } P_{KL}^i P_{K_1 L_1 }^j e_\mu ^L e_\nu ^{K_1 }  = T^{IJ} \varepsilon ^{\mu \nu \sigma _1 \rho _1 } \varepsilon _{ij\ell } P_{IL}^i P_{K_1 J}^j e_\mu ^L e_\nu ^{K_1 }  \\
& = T^{IJ} \varepsilon ^{\mu \nu \sigma _1 \rho _1 } \varepsilon _{ij\ell } P_{IL}^i P_{K_1 J}^j   \left(  P_n^{LK_1 } \Sigma _{\mu \nu }^n  + \bar P_n^{LK_1 } \bar \Sigma _{\mu \nu }^n\right),
\end{split}
\end{equation*}
where we used
\[
e_{[\mu }^L e_{\nu ]}^{K_1 }  = P_n^{LK_1 } \Sigma _{\mu \nu }^n  + \bar P_n^{LK_1 } \bar \Sigma _{\mu \nu }^n .
\]
\\
We use the property of the self-dual projection 
\begin{equation}\label{eq:79}
 -\frac{1}{2} \varepsilon _{ij\ell } P_{IJ}^i P_{KL}^j  =\frac{1}{4}\left( \eta _{JK} P_{\ell IL}- \eta _{JL} P_{\ell IK}\right)- \frac{1}{4}\left( {I \leftrightarrow J} \right),
\end{equation}
which can be easily checked when $I=0$ and $J, K, L$ are spatial indices, and when $I=K=0$ and $J, L$ are spatial indices, so the $SO(3, 1) \times SO(3,\mathbb{C})$ invariance asserts that this property is also satisfied when $I, J, K, L$ are all spatial indices. By using this property, we obtain
\begin{equation*}
\begin{split}
2I_1 &= T^{IJ} \varepsilon ^{\mu \nu \sigma _1 \rho _1 } \frac{-1}{2}\left( {\eta _{LK_1 } P_{\ell IJ}  - \eta _{LJ} P_{\ell IK_1 }  - \eta _{IK_1 } P_{\ell LJ}  + \eta _{IJ} P_{\ell LK_1 } } \right)e_\mu ^L e_\nu ^{K_1 }  \\ 
 & = T^{IJ} \varepsilon ^{\mu \nu \sigma _1 \rho _1 } \frac{-1}{2}\left( {P_{\ell IJ} g_{\mu \nu }  - \eta _{LJ} P_{\ell IK_1 } e_\mu ^L e_\nu ^{K_1 }  - \eta _{IK_1 } P_{\ell LJ} e_\mu ^L e_\nu ^{K_1 }  + \eta _{IJ} \Sigma _{\ell \mu \nu } } \right) \\ 
 & = T^{IJ} \varepsilon ^{\mu \nu \sigma _1 \rho _1 } \frac{-1}{2}\left(  - \eta _{LJ} P_{\ell IK_1 } e_\mu ^L e_\nu ^{K_1 }  - \eta _{IK_1 } P_{\ell LJ} e_\mu ^L e_\nu ^{K_1 }  + \eta _{IJ} \Sigma _{\ell \mu \nu }  \right)\\
 & = \frac{1}{2}T^{IJ} \varepsilon ^{\mu \nu \sigma _1 \rho _1 } \left(   \eta _{LJ} P_{\ell IK_1 }   +\eta _{IK_1 } P_{\ell LJ}   \right)e_{[\mu} ^L e_{\nu]} ^{K_1 }   - \frac{1}{2}\eta _{IJ} T^{IJ} \varepsilon ^{\mu \nu \sigma _1 \rho _1 } \Sigma _{\ell \mu \nu } ,
\end{split}
\end{equation*}
hence
\begin{equation*}
\begin{split}
  2I_1&=   \frac{1}{2}T^{IJ} \varepsilon ^{\mu \nu \sigma _1 \rho _1 } \left( {\eta _{LJ} P_{\ell IK_1 }  + \eta _{IK_1 } P_{\ell LJ} } \right)\left( {P_i^{LK_1 } \Sigma _{\mu \nu }^i  + \bar P_i^{LK_1 } \bar \Sigma _{\mu \nu }^i } \right) \\
& - \frac{1}{2}\eta _{IJ} T^{IJ} \varepsilon ^{\mu \nu \sigma _1 \rho _1 } \Sigma _{\ell \mu \nu }  .
\end{split}
\end{equation*}
Finally
\begin{equation*}
\begin{split}
 I_1&=   \frac{1}{4}T^{IJ} \varepsilon ^{\mu \nu \sigma _1 \rho _1 } \left( {\eta _{LJ} P_{\ell IK_1 }  + \eta _{IK_1 } P_{\ell LJ} } \right)\left( {P_i^{LK_1 } \Sigma _{\mu \nu }^i  + \bar P_i^{LK_1 } \bar \Sigma _{\mu \nu }^i } \right) \\
& - \frac{1}{4}\eta _{IJ} T^{IJ} \varepsilon ^{\mu \nu \sigma _1 \rho _1 } \Sigma _{\ell \mu \nu }  .
\end{split}
\end{equation*}
\\
And the second term of (\ref{eq:61}) becomes
\begin{equation*}
\begin{split}
  2I_2&=- T^{\mu _1 \nu _1 } \varepsilon ^{\mu \nu \rho \sigma _1 } \varepsilon _{ij\ell } \Sigma _{\mu _1 \mu }^i \Sigma _{\nu \rho }^j  {\delta _{\nu _1 }^{\rho _1 } }=  - T^{IJ} e_I^{\mu _1 } e_J^{\rho _1 } \varepsilon ^{\mu \nu \rho \sigma _1 } \varepsilon _{ij\ell } P_{KL}^i P_{K_1 L_1 }^j e_{\mu _1 }^K e_\mu ^L e_\nu ^{K_1 } e_\rho ^{L_1 }  \\
&=  - T^{IJ} \delta _I^K e_J^{\rho _1 } \varepsilon ^{\mu \nu \rho \sigma _1 } \varepsilon _{ij\ell } P_{KL}^i P_{K_1 L_1 }^j e_\mu ^L e_\nu ^{K_1 } e_\rho ^{L_1 } =  - T^{IJ} e_J^{\rho _1 } \varepsilon ^{\mu \nu \rho \sigma _1 } \varepsilon _{ij\ell } P_{IL}^i P_{K_1 L_1 }^j e_\mu ^L e_\nu ^{K_1 } e_\rho ^{L_1 }   \\ 
&=  - T^{IJ} \varepsilon ^{\mu \nu \rho \sigma _1 } \frac{-1}{2}\left( {\eta _{LK_1 } P_{\ell IL_1 }  - \eta _{LL_1 } P_{\ell IK_1 }  - \eta _{IK_1 } P_{\ell LL_1 }  + \eta _{IL_1 } P_{\ell LK_1 } } \right)e_J^{\rho _1 } e_\mu ^L e_\nu ^{K_1 } e_\rho ^{L_1 } = \\ 
     \frac{1}{2}T^{IJ}& \varepsilon ^{\mu \nu \rho \sigma _1 } \left( {P_{\ell IL_1 } g_{\mu \nu } e_\rho ^{L_1 }  - P_{\ell IK_1 } g_{\mu \rho } e_\nu ^{K_1 }  - \eta _{IK_1 } P_{\ell LL_1 } e_\mu ^L e_\nu ^{K_1 } e_\rho ^{L_1 }  + \eta _{IL_1 } P_{\ell LK_1 } e_\mu ^L e_\nu ^{K_1 } e_\rho ^{L_1 } } \right)e_J^{\rho _1 }  \\ 
 & =  \frac{1}{2}T^{IJ} \varepsilon ^{\mu \nu \rho \sigma _1 } \left( { - \eta _{IK_1 } P_{\ell LL_1 }  + \eta _{IL_1 } P_{\ell LK_1 }  } \right)e_J^{\rho _1 } e_\mu ^L e_\nu ^{K_1 } e_\rho ^{L_1 }.
\end{split}
\end{equation*}
Using $\varepsilon ^{\mu \nu \rho \sigma _1 } e_\mu ^L e_\nu ^{K_1 } e_\rho ^{L_1 }  = e\varepsilon ^{LK_1 L_1 M} e_M^{\sigma _1 } $, where $e$ is the determinant of $(e^I_\mu)$, we obtain
\begin{equation*}
\begin{split}
 2 I_2&  =   \frac{1}{2}eT^{IJ} \left( { - \eta _{IK_1 } P_{\ell LL_1 }  + \eta _{IL_1 } P_{\ell LK_1 } } \right)e_J^{\rho _1 } \varepsilon ^{LK_1 L_1 M} e_M^{\sigma _1 } \\
&=  \frac{1}{2}eT^{IJ} \left( { - \eta _{IK_1 } P_{\ell LL_1 }  + \eta _{IL_1 } P_{\ell LK_1 } } \right)\varepsilon ^{LK_1 L_1 M} e_J^{[\rho _1 } e_M^{\sigma _1 ]}  \\
&=   \frac{1}{2}eT^{IJ} \left( { - \eta _{IK_1 } P_{\ell LL_1 }  + \eta _{IL_1 } P_{\ell LK_1 } } \right)\varepsilon ^{LK_1 L_1 M} e_{[J}^{\rho _1 } e_{M]}^{\sigma _1 }  \\ 
 & =   \frac{1}{2}eT^{IJ} \left( { - \eta _{IK_1 } P_{\ell LL_1 }  + \eta _{IL_1 } P_{\ell LK_1 } } \right)\varepsilon ^{LK_1 L_1 M} \left( {P_{iJM} \Sigma ^{i\rho _1 \sigma _1 }  + \bar P_{iJM} \bar \Sigma ^{i\rho _1 \sigma _1 } } \right) .
\end{split}
\end{equation*}
We use the selfdual property
\[
P_{IJ}^i \varepsilon ^{IJKL}  = -2iP^{iKL}, 
\]
to get
\[
2I_2=  -i eT^{IJ} \left( {\eta _{IK_1 } P_\ell ^{K_1 M}  + \eta _{IL_1 } P_\ell ^{L_1 M} } \right)\left( {P_{iJM} \Sigma ^{i\rho _1 \sigma _1 }  + \bar P_{iJM} \bar \Sigma ^{i\rho _1 \sigma _1 } } \right).
\]
And using
\[
\frac{1}{{2!}}e^{-1} \varepsilon^{\mu \nu \rho \sigma } \Sigma _{\rho \sigma }^i    =\left( {*\Sigma ^i } \right)^{\mu \nu }  = \left( { - i\Sigma ^i } \right)^{\mu \nu }  =  - i\Sigma ^{\mu \nu i }, \quad e=\sqrt{-g},
\]
we obtain
\[
I_2  = -\frac{1}{4}T^{IJ} \left( {\eta _{IK_1 } P_\ell ^{K_1 M}  + \eta _{IL_1 } P_\ell ^{L_1 M} } \right)\left( { - P_{iJM} \varepsilon ^{\rho _1 \sigma _1 \rho \sigma } \Sigma _{\rho \sigma }^i  + \bar P_{iJM} \varepsilon ^{\rho _1 \sigma _1 \rho \sigma } \bar\Sigma _{\rho \sigma }^i } \right).
\]
\\
We get the third term of (\ref{eq:61}) by the replacing $\sigma _1  \leftrightarrow \rho _1 $ in $I_2$ with reversing its sign, we obtain
\[
I_3  = +\frac{1}{4}T^{IJ} \left( {\eta _{IK_1 } P_\ell ^{K_1 M}  + \eta _{IL_1 } P_\ell ^{L_1 M} } \right)\left( { - P_{iJM} \varepsilon ^{ \sigma _1 \rho _1 \rho \sigma } \Sigma _{\rho \sigma }^i  + \bar P_{iJM} \varepsilon ^{ \sigma _1 \rho _1 \rho \sigma } \bar\Sigma _{\rho \sigma }^i } \right)=I_2,
\]
therefore
\[
I_2 + I_3  = +\frac{1}{2}T^{IJ} \left( {\eta _{IK_1 } P_\ell ^{K_1 M}  + \eta _{IL_1 } P_\ell ^{L_1 M} } \right)\left( { - P_{iJM} \varepsilon ^{ \sigma _1 \rho _1 \rho \sigma } \Sigma _{\rho \sigma }^i  + \bar P_{iJM} \varepsilon ^{ \sigma _1 \rho _1 \rho \sigma } \bar\Sigma _{\rho \sigma }^i } \right).
\]
\\
The fourth term of (\ref{eq:61}) is
\[
I_4 = \frac{i}{6}T^{\mu _1 \nu _1 }g_{\mu _1 \nu _1 }  {\varepsilon ^{\mu \nu \sigma _1 \rho _1 }  \Sigma _{\mu \nu \ell}} .
\]
Using $T^{\mu  \nu  } g_{\mu  \nu  }= T^{IJ } \eta_{IJ }=T$, the equation (\ref{eq:61}) becomes
\begin{equation*}
\begin{split}
& \left. {\frac{{\delta S_m }}{{\delta B_{\sigma _1 \rho _1 }^\ell  }}} \right|_{constriant}  = I_1  + I_2  + I_3  + I_4\\ 
 & = \frac{1}{4}T^{IJ} \left( {\eta _{LJ} P_{\ell IK }  + \eta _{IK } P_{\ell LJ} } \right)\left( {P_i^{LK } \varepsilon ^{\mu \nu \sigma _1 \rho _1 } \Sigma _{\mu \nu }^i  + \bar P_i^{LK } \varepsilon ^{\mu \nu \sigma _1 \rho _1 } \bar \Sigma _{\mu \nu }^i ,} \right) - \frac{1}{4} T \varepsilon ^{\mu \nu \sigma _1 \rho _1 } \Sigma _{\ell \mu \nu }  \\ 
 & + \frac{1}{2}T^{IJ} \left( {\eta _{IK } P_\ell ^{K M}  + \eta _{IL } P_\ell ^{L M} } \right)\left( { - P_{iJM} \varepsilon ^{ \sigma _1 \rho _1 \rho \sigma } \Sigma _{\rho \sigma }^i  + \bar P_{iJM} \varepsilon ^{ \sigma _1 \rho _1 \rho \sigma } \bar\Sigma _{\rho \sigma }^i } \right) \\ 
 & + \frac{i}{6}T \varepsilon ^{\mu \nu \sigma _1 \rho _1 } \Sigma _{\mu \nu \ell } .
\end{split}
\end{equation*}
\\
We write this for short as
\[
\left. {\frac{{\delta S_m }}{{\delta B_{\sigma _1 \rho _1 }^\ell  }}} \right|_{constriant}  = \varepsilon ^{\mu \nu \sigma _1 \rho _1 } T_{\ell i} \Sigma _{\mu \nu }^i  + \varepsilon ^{\mu \nu \sigma _1 \rho _1 } \xi _{\ell i} \bar \Sigma _{\mu \nu }^i ,
\]
or
\begin{equation}\label{eq:180}
\left. {\frac{{\delta S_m }}{{\delta B^i  }}} \right|_{constriant}  ={\frac{{\delta S_m }}{{\delta \Sigma^i  }}}=  T_{ij} \Sigma^j  + \xi_{ij} \bar \Sigma ^j ,
\end{equation}
the complex matrices $T$ and $\xi$ are given by 
\begin{equation}\label{eq:a13}
\begin{split}
 T_{ij}  &= \frac{1}{4}T^{IJ} \left( {\eta _{LJ} P_{i IK}  + \eta _{IK} P_{i LJ} } \right)P_j^{LK}  -\frac{1}{2}T^{IJ} \left( {\eta _{IK} P_i ^{KM}  + \eta _{IL} P_i ^{LM} } \right)P_{jJM} \\
&\quad \quad+ T\left( {\frac{i}{6} - \frac{1}{4}} \right)\delta _{ij}  \\ 
  &= -  \frac{1}{2}\left( P_{iIK} P_{jJ} {^K}\right)  T^{IJ} + \left( {\frac{i}{6} - \frac{1}{4}} \right)T\delta _{ij},  
\end{split}
\end{equation}
and
\begin{equation*}
\begin{split}
 \xi_{ij}  &= \frac{1}{4}T^{IJ} \left( {\eta _{LJ} P_{i IK}  + \eta _{IK} P_{i LJ} } \right)\bar P_j^{LK}  + \frac{1}{2}T^{IJ} \left( {\eta _{IK} P_i ^{KM}  + \eta _{IL} P_i ^{LM} } \right)\bar P_{jJM}  \\ 
  &= \frac{3}{2} \left(  P_{iIK} \bar P_{jJ} {^K} \right) T^{IJ}. 
\end{split}
\end{equation*}
The self-dual projection matrices $P_{ IJ}^i$ are given in the equation (\ref{eq:selfd}),
\begin{equation}\label{eq:a12}
P^i_{IJ}=\frac{1}{2}\varepsilon ^i{_{jk}}, \text{ for } I=i, J=j, \text{ and } P^i_{0j}=-P^i_{j0}=-\frac{i}{2}\delta^i_j, \text{ for } I=0, J=j \ne 0.
\end{equation}
and $\bar P_{ IJ}^i$ are their complex conjugate. We get
\begin{equation*}
\begin{split}
& \left( {P_{iIK} P_{jJ} ^K } \right)T^{IJ}  = \left( {P_{iI0} P_{jJ} ^0  + P_{iI\ell } P_{jJ} ^\ell  } \right)T^{IJ}  \\ 
&  = P_{i\ell 0} P_{jm} ^0 T^{\ell m}  + P_{i0\ell } P_{j0} ^\ell  T^{00}  + P_{im\ell } P_{j0} ^\ell  T^{m0}  + P_{i0\ell } P_{jm} ^\ell  T^{0m}  + P_{in\ell } P_{jm} ^\ell  T^{nm}  \\ 
&  =  - P_{i\ell 0} P_{jm0} T^{\ell m}  + P_{i0\ell } P_{j0} ^\ell  T^{00}  + P_{im\ell } P_{j0} ^\ell  T^{m0}  + P_{i0\ell } P_{jm} ^\ell  T^{0m}  + P_{in\ell } P_{jm} ^\ell  T^{nm}  \\
  &=  - \frac{i}{2}\delta _{i\ell } \frac{i}{2}\delta _{jm} T^{\ell m}  + \frac{{ - i}}{2}\delta _{i\ell } \frac{{ - i}}{2}\delta _j^\ell  T^{00}  + \frac{1}{2}\varepsilon _{im\ell } \frac{{ - i}}{2}\delta _j^\ell  T^{m0}  + \frac{{ - i}}{2}\delta _{i\ell } \frac{1}{2}\varepsilon _{jm} ^\ell  T^{0m}  + \frac{1}{2}\varepsilon _{in\ell } \frac{1}{2}\varepsilon _{jm} ^\ell  T^{nm}  \\ 
  &= \frac{1}{4}T_{ij}  - \frac{1}{4}\delta _{ij} T^{00}  + \frac{{ - i}}{4}\varepsilon _{imj} T^{m0}  + \frac{{ - i}}{4}\varepsilon _{jmi} T^{0m}  + \frac{1}{4}\left( {\delta _{ij} \delta _{nm}  - \delta _{im} \delta _{nj} } \right)T^{nm}   \\
&= \frac{1}{4}T_{ij}  - \frac{1}{4}\delta _{ij} T^{00}  + \frac{{ - i}}{4}\varepsilon _{imj} T^{0m}  + \frac{{ - i}}{4}\varepsilon _{jmi} T^{0m}  + \frac{1}{4}\delta _{ij} \delta _{nm} T^{nm}  - \frac{1}{4}T_{ji}  \\ 
&  =  - \frac{1}{4}\delta _{ij} T^{00}  + \frac{1}{4}\delta _{ij} \delta _{nm} T^{nm}  = \frac{1}{4}\delta _{ij} \eta _{00} T^{00}  + \frac{1}{4}\delta _{ij} \eta _{nm} T^{nm}  = \frac{1}{4}\delta _{ij} \eta _{IJ} T^{IJ}  = \frac{1}{4}\delta _{ij} T. 
\end{split}
\end{equation*}
Using this in the equation (\ref{eq:a13}), we get
\begin{equation}\label{eq:a14}
T_{ij}= - \frac{1}{8}\delta _{ij} T + \left( {\frac{i}{6} - \frac{1}{4}} \right)T\delta _{ij}=\left( {\frac{i}{6} - \frac{3}{8}} \right)T\delta _{ij}. 
\end{equation}

\section{Appendix C}
$\rm{I}$- We verify that the (0,2) tensor field $\Sigma ^{\mu \nu i}$ defined in
\[
 - i\Sigma ^{\mu \nu i }=\frac{1}{{2!}}e^{-1} \varepsilon^{\mu \nu \rho \sigma } \Sigma _{\rho \sigma }^i , \quad  \text{for}\quad \varepsilon^{0123}=-\varepsilon_{0123}=1, \quad e=\det(e_\mu^I) ,
\]
is inversion of the 2-form $\Sigma_{\mu \nu }^i$, that is $\Sigma _{\mu \nu }^i  \Sigma ^{\mu \nu }_j=\delta^i_j$. Multiplying with $\Sigma _{\mu \nu j} $ and sum over contracted indices, we get
\[
 - i\Sigma _{\mu \nu j} \Sigma ^{\mu \nu i }=\frac{1}{{2}}e^{-1} \varepsilon^{\mu \nu \rho \sigma }\Sigma _{\mu \nu j} \Sigma _{\rho \sigma }^i .
\]
Then using $\Sigma^i_{\mu \nu }=P_{IJ}^i e^I_\mu e^J_\nu $, implies
\[
 - i\Sigma _{\mu \nu j} \Sigma ^{\mu \nu i }=\frac{1}{{2}}e^{-1} \varepsilon^{\mu \nu \rho \sigma }P_{jIJ} P^i_{KL} e^I_\mu e^J_\nu e^K_\rho e^L_\sigma.
\]
Using $\varepsilon^{\mu \nu \rho \sigma } e^I_\mu e^J_\nu e^K_\rho e^L_\sigma=e \varepsilon^{IJKL }$, to get
\[
 - i\Sigma _{\mu \nu j} \Sigma ^{\mu \nu i }=\frac{1}{{2}}e^{-1}P_{jIJ} P^i_{KL} e \varepsilon^{IJKL }=\frac{1}{{2}}P_{jIJ} P^i_{KL} \varepsilon^{IJKL }.
\]
Then we use the self-dual projection property
\[
P_{IJ}^i \varepsilon ^{IJKL}  = -2iP^{iKL}
\]
to obtain
\[
 - i\Sigma _{\mu \nu j} \Sigma ^{\mu \nu i }=\frac{1}{{2}}P_{jIJ}(-2iP^{iIJ}) =P_{jIJ}(-iP^{iIJ }).
\]
Therefore
\[
 \Sigma _{\mu \nu j} \Sigma ^{\mu \nu i } =P_{jIJ} P^{iIJ}=\delta^i_j,
\]
where we used the self-dual projection property $P_{jIJ} P^{iIJ}=\delta^i_j$. The sum is over the contracted indices.
\\

$\rm{II}$- The local Lorentz invariance gives a conserved current, call it spin current, or Lorentz current. In flat coordinates $(x^I)$, the spin current for arbitrary field $\varphi$ is (\cite{Mark}, section 22)
\begin{equation}\label{eq:z71}
M^{IJK}=x^J T^{IK} - x^K T^{IJ},
\end{equation}
the conservation law is $\partial_I M^{IJK}=0$, it follows from the conservation of energy-momentum tensor $\partial_I T^{IJ}=0$ of the field $\varphi$. When the field carries spinor indices, like $\varphi^\alpha$, this adds a term like
\[
M^{IJK}=x^J T^{IK} - x^K T^{IJ} + \pi_\alpha^I (S^{JK})^\alpha{_\beta} \varphi^\beta.
\]
In arbitrary coordinates $(x^\mu)$, we write $x^I=e^I_\mu x^\mu$, where $x^I$ becomes tangent vector and $e^I_\mu$ are gravitational fields, also $T^{\mu\nu}=e_I^\mu e_J^\nu T^{IJ}$. One can write
\begin{equation}\label{eq:z72}
M^{\mu IJ}=e^I_\nu x^\nu T^{\mu J} - e^J_\nu x^\nu T^{\mu I} .
\end{equation}
Therefore $D_\mu M^{\mu IJ}=0$ is satisfied when $D_\mu e^I_\nu=0$ and $D_\mu T^{\mu I}=0$. Here $e^I_\mu$ and $T^{\mu I}$ are functions of the coordinates $(x^\mu)$, so also $M^{\mu IJ}$ are functions of $(x^\mu)$. The current $M^{\mu IJ}$ couples to the spin connection $\omega^{IJ}_\mu$ for local symmetry of Lorentz group $SO(3,1)$. The equation $D_\mu e^I_\nu=0$ defines affine connection $\nabla$ on $TM\times TM$ and spin connection $\omega$ on $TM\times \mathfrak{so}(3, 1)$. When the coordinates $(x^\mu)$ are flat, so $e^I_\mu=\delta^I_\mu$, by that the equation (\ref{eq:z72}) becomes (\ref{eq:z71}).

\end{document}